\documentclass[12pt]{article}

\usepackage{amsmath}
\usepackage{paralist}
\usepackage[colorlinks=true]{hyperref}
\usepackage{amssymb}
\usepackage{amsthm}
\usepackage{latexsym}
\usepackage{cite}
\usepackage{psfrag}
\usepackage{color}
\usepackage{amsfonts}
\usepackage{mathrsfs}
\usepackage{url}
\usepackage{graphicx}
\hypersetup{urlcolor=blue, citecolor=red}

\makeatletter
\@addtoreset{equation}{section}
\makeatother

\newtheorem{theorem}{Theorem}[section]
\newtheorem{corollary}[theorem]{Corollary}

\newtheorem{proposition}[theorem]{Proposition}
\newtheorem{lemma}[theorem]{Lemma}
\theoremstyle{definition}
\newtheorem{definition}[theorem]{Definition}
\newtheorem{remark}[theorem]{Remark}
\newtheorem{notation}[theorem]{Notation}
\newtheorem{example}[theorem]{Example}

\newcommand{\A}{\mathcal{A}}

\newcommand{\D}{\mathcal{D}}

\newcommand{\K}{\mathcal{K}}

\newcommand{\Pc}{\mathcal{P}}

\newcommand{\T}{\mathcal{T}}
\newcommand{\V}{\mathcal{V}}

\newcommand{\PG}{\mathrm{PG}}

\newcommand{\F}{\mathbb{F}}
\newcommand{\U}{\mathscr{U}}

\newcommand{\G}{\widehat{\mathbb{G}}}

\begin{document}

\title{Upper bounds on the length function for covering codes with covering radius $R$ and codimension $tR+1$
\date{}
}
\maketitle
\begin{center}
{\sc Alexander A. Davydov
\footnote{A.A. Davydov ORCID \url{https://orcid.org/0000-0002-5827-4560}}
}\\
{\sc\small Institute for Information Transmission Problems (Kharkevich institute)}\\
 {\sc\small Russian Academy of Sciences}\\
 {\sc\small Moscow, 127051, Russian Federation}\\\emph {E-mail address:} adav@iitp.ru\medskip\\
 {\sc Stefano Marcugini
 \footnote{S. Marcugini ORCID \url{https://orcid.org/0000-0002-7961-0260}},
 Fernanda Pambianco
 \footnote{F. Pambianco ORCID \url{https://orcid.org/0000-0001-5476-5365}}
 }\\
 {\sc\small Department of  Mathematics  and Computer Science,  Perugia University,}\\
 {\sc\small Perugia, 06123, Italy}\\
 \emph{E-mail address:} \{stefano.marcugini, fernanda.pambianco\}@unipg.it
\end{center}

\textbf{Abstract.}
The length function $\ell_q(r,R)$ is the smallest
length of a $ q $-ary linear code with codimension (redundancy) $r$ and covering radius $R$. In this work, new upper bounds on $\ell_q(tR+1,R)$ are obtained in the following forms:
\begin{equation*}
\begin{split}
    &(a)~\ell_q(r,R)\le cq^{(r-R)/R}\cdot\sqrt[R]{\ln q},~  R\ge3,~r=tR+1,~t\ge1,\\
 &\phantom{(a)~} q\text{ is an arbitrary prime power},~c\text{ is independent of }q.
  \end{split}
\end{equation*}
\begin{equation*}
\begin{split}
    &(b)~\ell_q(r,R)< 3.43Rq^{(r-R)/R}\cdot\sqrt[R]{\ln q},~  R\ge3,~r=tR+1,~t\ge1,\\
 &\phantom{(b)~} q\text{ is an arbitrary prime power},~q\text{ is large enough}.
  \end{split}
\end{equation*}
In the literature, for $q=(q')^R$ with $q'$ a prime power, smaller upper bounds are known; however, when $q$ is an arbitrary prime power, the  bounds of this paper are better than the known ones.

For $t=1$, we use a one-to-one correspondence between $[n,n-(R+1)]_qR$ codes and $(R-1)$-saturating $n$-sets in the projective space $\mathrm{PG}(R,q)$. A new construction of such saturating sets providing sets of small size is proposed. Then the $[n,n-(R+1)]_qR$  codes, obtained by geometrical methods, are taken as the starting ones
in the lift-constructions (so-called ``$q^m$-concatenating constructions'') for covering codes to obtain infinite families of codes with growing codimension $r=tR+1$, $t\ge1$.



\section{Introduction}

Let $\F_{q}$ be the Galois field with $q$ elements. Let $\F_{q}^{\,n}$ be the $n$-dimensional vector space over~$\F_{q}.$ Denote by $[n,n-r,d]_q$ a $q$-ary linear code of length $n$, codimension (redundancy)~$r$, and minimum distance $d$. Usually, $d$ is omitted when not relevant. For an introduction to coding theory, see \cite{libroBierbrauer,HufPless,MWS,Roth}.

The sphere of radius $R$ with center $c$ in $F_{q}^{\,n}$
is the set $\{v:v\in F_{q}^{\,n},$ $d(v,c)\leq R\}$ where $d(v,c)$ is the Hamming distance between the vectors $v$ and $c$.

\begin{definition} \label{Def1_CoverRad}
 A linear $[n,n-r,d]_{q}$ code has \emph{covering radius} $R$ and is denoted as an $[n,n-r,d]_{q}R$ code if any of the following equivalent properties holds:

  \textbf{(i)} The value $R$  is the smallest integer such that the space $\F_{q}^{\,n}$ is covered by
the spheres of radius $R$ centered at the codewords.

\textbf{(ii)}
Every vector in $\F_{q}^{\,r}$ is equal to a linear combination of at most $R$ columns of a parity check matrix of the code, and $R$ is the smallest value with
this property.
\end{definition}

The covering density $\mu$  of an $[n, n-r,d]_qR$-code is defined as the ratio of the total volume of all spheres of radius $R$ centered at the codewords to the volume of the space~$\F_{q}^{\,n}$. By Definition \ref{Def1_CoverRad}(i), we have $\mu\ge1$.
The covering quality of a code is better if its covering density is smaller.
For the fixed parameters $q$, $r$, $R$, the covering density $\mu$ of an $[n, n-r]_{q}R$ code decreases with decreasing~$n$.

The covering problem for codes is that of finding codes with small covering radius
with respect to their lengths and dimensions. Codes investigated from the point
of view of the covering problem are usually called \emph{covering codes} (in contrast to
error-correcting codes). If covering radius and
codimension are fixed then the covering problem for codes is that of finding codes with small length and/or obtaining good upper bounds for the length.

\begin{definition}
The \emph{length function} $\ell_q(r,R)$ is the smallest
length of a $ q $-ary linear code with codimension (redundancy) $r$ and covering radius $R$.
\end{definition}

For an introduction to coverings of vector Hamming spaces over finite fields and covering codes, see
\cite{GrahSlo,Handbook-coverings,CHLS-bookCovCod,DGMP-AMC}, the references therein, and the online bibliography~\cite{LobstBibl}.

Studying covering codes is a classical combinatorial problem.
Covering codes are connected  with many areas of theory and practice, for example, with decoding errors
and erasures, data compression, write-once memories, football pools, Caley graphs,
and Berlekamp-Gale games, see \cite[Section
1.2]{CHLS-bookCovCod}. Codes of covering radius~2 and codimension
3 are relevant for the degree/diameter problem in graph theory
\cite {Cayley,KKKPS} and defining sets of block designs
\cite{Bo-Sz-Ti}. Covering
codes can also be used  in steganography \cite[Chapter
14]{libroBierbrauer}, \cite {BierbStegan,GalKaba,GalKabaPIT}, in databases
\cite{PartSumQuer}, in constructions of identifying codes
\cite{Identif,identif2}, for solving the so-called learning parity with noise (LPN) \cite{LPN_CovCod}, in an analysis of blocking switches
\cite{swithc}, in reduced representations of logic functions \cite{Bulev}, in the list decoding of error correcting codes \cite{ChenListDec}, in cryptography \cite{MenNewCrypCovCod}. There are connections between covering
codes and a popular game puzzle, called ``Hats-on-a-line'' \cite{Hats-on-line,LenSerHat2002}.

Let $\PG(N,q)$ be the $N$-dimensional projective space over the Galois field~$\mathbb{F}_q$.
 We will say ``$N$-space'' (or ``$M$-subspace'') when the value of $q$ is clear by the context.

We say that $M$ points of $\PG(N,q)$ are \emph{in general position} if they are not
contained in an $(M-2)$-subspace. In particular, $N+1$ points of $\PG(N,q)$ are in general positions if and only if they do not belong to the same hyperplane. A point of $\PG(N,q)$ in homogeneous coordinates can be considered as a vector of $\F_q^{N+1}$. In this case, points in general position correspond to linear independent vectors.

Effective methods to obtain upper bounds on the length function $\ell_q(r,R)$ are connected with \emph{saturating sets in $\PG(N,q)$}.

\begin{definition}\label{def1_usual satur}
 A point set $S\subseteq\PG(N,q)$ is
$\rho$-\emph{saturating} if any of the following equivalent properties holds:

\textbf{(i)}  For any point $A\in\PG(N,q)$ there exists a value $\overline{\rho}\le\rho$ such that in $S$ there are
$\overline{\rho}+1$ points in general position generating a $\overline{\rho}$-subspace of $\PG(N,q)$ in which $A$ lies, and $\rho$ is the smallest value with this property.

\textbf{(ii)} Every
point $A\in\PG(N,q)$ can be written as a linear combination of at most $\rho+1$ points of $S$, and $\rho$ is the smallest value with this property (cf. Definition~\ref{Def1_CoverRad}(ii)).
\end{definition}

 In the literature, saturating sets are also called ``saturated
sets'', ``spanning sets'', and ``dense sets''.

Let $s_q(N,\rho)$ be the smallest size of a $\rho$-saturating set in $\mathrm{PG}(N,q)$.

If $q$-ary positions of a column of an $r\times n$ parity check matrix of an $[n,n-r]_qR$ code are treated as homogeneous coordinates of a point in $\mathrm{PG}(r-1,q)$ then this parity check matrix defines an $(R-1)$-saturating set of size $n$ in $\mathrm{PG}(r-1,q)$, and vice versa. So, there is a \emph{one-to-one correspondence} between $[n,n-r]_qR$ codes and $(R-1)$-saturating $n$-sets in $\PG(r-1,q)$. Therefore,
\begin{equation}\label{eq1:one-to-one}
    \ell_q(r,R)=s_q(r-1,R-1).
\end{equation}

For an introduction to the projective geometry over finite fields and its connection with coding theory, see \cite{EtzStorm2016,Giul2013Survey,Hirs,Hirs_PG3q,HirsStor-2001,HirsThas-2015,Klein-Stor,LandSt,DGMP-AMC} and the references therein. Note also that in the papers \cite{BDGMP-R2Redun2016,BDGMP-R2Castle,BDGMP-R2R3CC_2019,BDMP-arXivR3_2018,Dav90PIT,Dav95,DGMP-Bulg2008,DGMP-AMC,DMP-R=tR2019,DMP-R=3Redun2019,DMP-ICCSA2020,%
DavOst-EJC2000,DavOst-DESI2010,DavOst-IEEE2001,denaux,denaux2,Giul2013Survey,GrahSlo,HegNagy,Nagy,Struik}, distinct aspects of covering codes and saturating sets, including their joint investigations,  are considered.

Throughout the paper, $c$ and $c_i$ are constants independent of $q$ but it is possible that $c$ and $c_i$ are dependent on $r$ and $R$.

In \cite{BDGMP-R2R3CC_2019,DMP-R=3Redun2019}, \cite[Proposition 4.2.1]{denaux}, see also the references therein, the following lower bound is considered:
\begin{equation}\label{eq1:lowbound}
\ell_q(r,R)=s_q(r-1,R-1)\ge cq^{(r-R)/R},~  R\text{ and }r\text{ fixed}.
\end{equation}
In \cite{DGMP-Bulg2008,DGMP-AMC}, see also the references therein, the bound \eqref{eq1:lowbound} is given in another \mbox{(asymptotic)} form.

%In future, when we say "$q$ has arbitrary structure" we have in mind that the structure of $q$ is arbitrary, for example, other than in \eqref{eq1:spec_cases1}, \eqref{eq1:spec_cases2}.

Let $t,s$ be integer. Let $q'$ be a prime power. In the literature, it is proved that in the following cases, the bound \eqref{eq1:lowbound} is achieved:
 \begin{equation}\label{eq1:cases_lowbnd}
 \begin{split}
 &r\ne tR,~  q=(q')^R,~ \text{\cite{DGMP-Bulg2008,DGMP-AMC,denaux,denaux2,HegNagy,KKKPS}}; \\
 &R=sR',~ r=tR+s, ~q=(q')^{R'},~\text{\cite{DGMP-Bulg2008,DGMP-AMC}};\\
 &r=tR,~q\text{ \emph{is an arbitrary prime power}},~ \text{\cite{Dav95,DavOst-EJC2000,DavOst-IEEE2001,DGMP-Bulg2008,DGMP-AMC,DMP-R=tR2019}}.
 \end{split}
 \end{equation}

In the general case, for arbitrary $r,R, q$, in particular when $r\ne tR$ and $q$ is an arbitrary prime power, the problem
of achieving the bound \eqref{eq1:lowbound} is open.

In the literature \cite{BDGMP-R2Redun2016,BDGMP-R2Castle,BDGMP-R2R3CC_2019,BDMP-arXivR3_2018,DMP-R=3Redun2019,DMP-ICCSA2020,Bo-Sz-Ti,Kovacs,Nagy}, for $R=2$ with any $q$ and $R=3$ with $q$ upper bounded, upper bounds of the following form are obtained:
\begin{equation}\label{eq1:upbound}
\begin{split}
    &\ell_q(r,R)=s_q(r-1,R-1)\le cq^{(r-R)/R}\cdot\sqrt[R]{\ln q},\\
 &r\ne tR,~ q\text{ is an arbitrary prime power},
  \end{split}
\end{equation}
see Section \ref{subsec:knownR2R3} for details.
\begin{remark}\label{rem1:price}
In the bounds of the form \eqref{eq1:upbound}, the ``price'' of the non-restrict structure of $q$ is the factor $\sqrt[R]{\ln q}$.
\end{remark}

For $R\ge3$, $r\ne tR$, when $q$ is an arbitrary prime power, the standard known way to obtain upper bounds on the length function is the so-called
direct sum construction \cite{CHLS-bookCovCod,Handbook-coverings}. This construction gives, see Section~\ref{subsec:dirSum},
\begin{equation}\label{eq1:DS}
\begin{split}
    &\ell_q(tR+1,R)\le cq^{(r-R)/R+(R-2)/2R}\sqrt{\ln q},~  R\ge3,~r=tR+1,~t\ge1,\\
 & q\text{ is an arbitrary prime power},
  \end{split}
\end{equation}
that is worse than the bound \eqref{eq1:upbound}.

In this paper, for $q$  an arbitrary prime power, we prove the upper bound of the form \eqref{eq1:upbound} on the length function $\ell_q(Rt+1,R)$, $t\ge1$, $R\ge3$, see Section \ref{sec:main_res}. In general, we obtain the bounds \eqref{eq1:new} and \eqref{eq1:newinf}.
\begin{equation}\label{eq1:new}
\begin{split}
    &\ell_q(tR+1,R)\le cq^{(r-R)/R}\cdot\sqrt[R]{\ln q},~  R\ge3,~r=tR+1,~t\ge1,\\
 & q\text{ is an arbitrary prime power}.
  \end{split}
\end{equation}
\begin{equation}\label{eq1:newinf}
\begin{split}
    &\ell_q(r,R)< 3.43Rq^{(r-R)/R}\cdot\sqrt[R]{\ln q},~  R\ge3,~r=tR+1,~t\ge1,\\
 & q\text{ is an arbitrary prime power},~q\text{ is large enough}.
  \end{split}
\end{equation}
The bounds obtained are new and essentially better than the known ones of the form~\eqref{eq1:DS}.

The main contribution in our approach is obtaining new upper bounds on\linebreak
 $\ell_q(R+1,R)$ by a geometric way. Then we use the lift-constructions for covering codes to obtain the bounds on $\ell_q(tR+1,R)$.

In the beginning we consider the case $t=1$ in projective geometry language. We prove the upper bound on the smallest size $s_q(R,R-1)$ of $(R-1)$-saturating sets in $\PG(R,q)$. For it we propose Construction A, that obtains a saturating set by a step-by-step algorithm.
Then we estimate the size of the obtained $n$-set that corresponds to an $[n,n-(R+1)]_qR$ code. This gives the bounds on $\ell_q(R+1,R)$.

For $t\ge2$, we use the lift-constructions for covering codes.  These constructions are variants of the so-called ``$q^m$-concatena\-ting constructions'' proposed in \cite{Dav90PIT} and developed in \cite{Dav95,DGMP-Bulg2008,DGMP-AMC,DavOst-IEEE2001,DavOst-DESI2010}, see also the references therein and \cite[Section~5.4]{CHLS-bookCovCod}. The $q^m$-concatenating constructions obtain infinite families of covering codes with growing codimension using a starting code with a small one. The covering density of the codes from the infinite familes is approximately the same as for the starting code.

We take the $[n,n-(R+1)]_qR$ codes corresponding to the $(R-1)$-saturating sets in $\PG(R,q)$ as the starting ones for the $q^m$-concatena\-ting constructions and obtain infinite families of covering codes with growing codimension $r=tR+1$, $t\ge1$. These families provide the bounds on $\ell_q(tR+1,R)$.

 The paper is organized as follows. In Section \ref{sec_known}, the known upper bounds for $r\ne tR$ and arbitrary prime power $q$ are given. In Section \ref{sec:main_res}, the main new results are written. Section \ref{sec:constrA} describes Construction A that obtains $(R-1)$-saturating $n$-sets in $\PG(R,q)$ corresponding to $[n,n-(R+1)]_qR$ codes. Estimates of sizes of saturating sets obtained by  Construction~A and the corresponding upper bounds are given in Sections \ref{sec:est_size} and \ref{sec:boundsR_R-1}. In Section~\ref{sec:R=3}, for illustration, bounds on the length function $\ell_q(4,3)$ and 2-saturating sets in $\PG(3,q)$ are considered.  In Section \ref{sec:qmconc}, upper bounds on the length function $\ell_q(tR+1,R)$ are obtained for growing $t\ge1$. These bounds are provided by infinite families of covering codes with growing codimension $r=tR+1$, $t\ge1$, created by the $q^m$-concatenating constructions.

 \section{The known upper bounds for $r\ne tR$ and arbitrary prime power $q$}\label{sec_known}

\subsection{Bounds for $R=2$ with any $q$ and $R=3$ with $q$ upper bounded}\label{subsec:knownR2R3}

Let $\delta_{i,j}$ be the Kronecker delta. For $R=2,3$, $r\ne tR$, when $q$ is an arbitrary prime power, as far as it is known to the authors, the best upper bounds in the literature  are as follows.
\begin{equation}
\begin{split}
&\ell_q(r,2)=s_q(r-1,1)\le\Phi(q)\cdot q^{(r-2)/2}\cdot\sqrt{\ln q}+2\lfloor q^{(r-5)/2}\rfloor,~\text{\cite{BDGMP-R2Castle,BDGMP-R2R3CC_2019,Nagy}}\label{eq2:r_odd_R=2}\\
&r=2t+1\geq3,~r\neq9,13,~t\ge1,~t\ne4,6;\\
&\Phi(q)=\left\{\begin{array}{@{}ll@{}}
    0.998\sqrt{3}<1.729&\text{if }q\le160001\smallskip\\
    1.05\sqrt{3}<1.819&\text{if }160001<q\le321007\smallskip\\
    \sqrt{3+\frac{\ln\ln q}{\ln q}}+\sqrt{\frac{1}{3\ln^2 q}}+\frac{3}{\sqrt{q\ln q}}<1.836&\text{if }q>321007
    \end{array}
    \right.;\\
    &\lim_{q\rightarrow\infty}\Phi(q)=\sqrt{3}.
\end{split}
    \end{equation}
    \begin{equation}\label{eq2:R=3_r=4}
    \begin{split}
  &\ell_q(r,3)=s_q(r-1,2)<c_4\cdot q^{(r-3)/3}\cdot\sqrt[3]{\ln q}+3\lfloor q^{(r-7)/3}\rfloor~\text{\cite{BDGMP-R2R3CC_2019,DMP-R=3Redun2019,DMP-ICCSA2020}}\\
  &+2\left\lfloor q^{(r-10)/3}\right\rfloor+\delta_{r,13},~r=3t+1,~t\ge1;~c_4< \left\{\begin{array}{@{}l@{}}
  2.61 \text{ if }  13\le q\le4373 \smallskip\\
  2.65\text{ if } 4373< q\le7057
\end{array}
  \right..
\end{split}
\end{equation}
    \begin{equation}\label{eq2:R=3_r=5}
    \begin{split}
  &\ell_q(r,3)=s_q(r-1,2)<c_5\cdot q^{(r-3)/3}\cdot\sqrt[3]{\ln q}+3\lfloor q^{(r-8)/3}\rfloor~~\text{\cite{BDGMP-R2R3CC_2019,DMP-R=3Redun2019}}\\
  &+2\left\lfloor q^{(r-11)/3}\right\rfloor+\delta_{r,14},~
r=3t+2,~ t\ge1;~c_5< \left\{
\begin{array}{@{}l@{}}
  2.785 \text{ if }  11\le q\le401 \smallskip\\
  2.884\text{ if } 401< q\le839
\end{array}
  \right..
    \end{split}
    \end{equation}

In \eqref{eq2:r_odd_R=2}, the results  for $r=3$ are obtained by computer search, if $q\le 321007$, and in a theoretical way, if $q> 321007$.
In \eqref{eq2:R=3_r=4}, \eqref{eq2:R=3_r=5}, the results  for $r=4,5$,  are obtained by computer search.
The rest of the results in \eqref{eq2:r_odd_R=2}--\eqref{eq2:R=3_r=5} are obtained by applying the lift-construct\-ions ($q^m$-concatenating constructions) for covering codes \cite[Section~5.4]{CHLS-bookCovCod}, \cite{Dav90PIT,Dav95,DGMP-Bulg2008,DGMP-AMC,DavOst-IEEE2001,DavOst-DESI2010}.

\subsection{Direct sum construction}\label{subsec:dirSum}

The direct sum construction \cite{CHLS-bookCovCod,Handbook-coverings,DGMP-AMC} forms an $[n_1+n_2,n_1+n_2-(r_1+r_2)]_qR$ code $V$ with $R = R_1+R_2$ from
two codes: an $[n_1,n_1-r_1]_qR_1$ code $V_1$ and an $[n_2,n_2-r_2]_qR_2$ code $V_2$.

For example, for the code $V$, let $R=3$, $r=3t+1$. Choose $r_1=2t+1$, $R_1=2$, $n_1\thickapprox c_1q^{(r_1-2)/2}\sqrt{\ln q}=c_1q^{t-1}\sqrt{q\ln q}$, see \eqref{eq1:upbound}, \eqref{eq2:r_odd_R=2}; $r_2=t$, $R_2=1$, $n_2=(q^t-1)/(q-1)\thickapprox c_2q^{t-1}$, i.e. $V_2$ is the $[\frac{q^t-1}{q-1},\frac{q^t-1}{q-1}-t]_q1$ Hamming code. The length $n$ of the resulting code $V$ is $n\thickapprox cq^{t-1}\sqrt{q\ln q}=cq^{(r-R)/R+(R-2)/2R}\sqrt{\ln q}$.

Similarly, one can show that \eqref{eq1:DS} holds. For the code $V$,  let $r=tR+1$. Choose $r_1=2t+1$, $R_1=2$, $n_1\thickapprox c_1q^{t-1}\sqrt{q\ln q}$; $r_2=(R-2)t$, $R_2=R-2$, $n_2=(R-2)(q^t-1)/(q-1)\thickapprox c_2q^{t-1}$, i.e. $V_2$ is the sequential direct sum of $R-2$ Hamming codes. Again, the length $n$ of $V$ is $n\thickapprox cq^{t-1}\sqrt{q\ln q}=cq^{(r-R)/R+(R-2)/2R}\sqrt{\ln q}$.

\section{The main new results}\label{sec:main_res}

\begin{notation}\label{not1}
Throughout the paper, fixed $R$, we denote the following:
\begin{itemize}
  \item $\theta_{R,q}=(q^{R+1}-1)/(q-1)$ is the number of points in the projective space $\PG(R,q)$.

  \item $\triangleq$ is the sign ``equality by definition''.

  \item $\lambda>0$ is a positive constant independent of $q$ and $R$, its value can be assigned arbitrarily.

  \item $D^{\min}_R$ is a constant independent of $q$ and $\lambda$ and dependent on $R$.

  \item $Q_{\lambda,R}$, $C_{\lambda,R}$, and $D_{\lambda,R}$,  are constants independent of $q$ and dependent on $\lambda$ and~$R$.

  \item  $\beta_{\lambda,R}(q)$, $\Upsilon_{\lambda,R}(q)$, and $\Omega_{\lambda,R}(q)$ are functions of $q$, parameters of which are dependent on $\lambda$ and~$R$.

  \item
  \begin{equation}\label{eq3:constantDmin}
D^{\min}_R\triangleq\frac{R}{R-1}\sqrt[R]{R(R-1)\cdot R!}.
\end{equation}

  \item
  \begin{equation}\label{eq3:constantDlamb}
D_{\lambda,R}\triangleq \lambda+\frac{R\cdot R!}{\lambda^{R-1}}.
\end{equation}

  \item
\begin{equation}\label{eq3:beta0_def}
\beta_{\lambda,R}(q)\triangleq\lambda-\frac{R-1}{\sqrt[R]{q\ln q}}.
\end{equation}

  \item
\begin{equation}\label{eq3:Upsilon_def}
\Upsilon_{\lambda,R}(q)\triangleq\frac{\lambda^{R-1}}{(R-1)!}\sqrt[R]{\frac{\ln^{R-1} q}{q}}.
\end{equation}

  \item
  \begin{equation}\label{eq3:Omega_def}
\Omega_{\lambda,R}(q)\triangleq\lambda+\frac{R\cdot R!}{\beta_{\lambda,R}^{R-1}(q)}\cdot\frac{2}{2-\frac{1}{q}-\Upsilon_{\lambda,R}(q)}.
\end{equation}

  \item
\begin{equation}\label{eq3:Qlx}
  Q_{\lambda,R}\triangleq\lceil x\rceil,~x\triangleq\left\{\begin{array}{ccc}
                         e^{R-1} & \text{ if } & \Upsilon_{\lambda,R}(e^{R-1}) \le 1\\
                         y  & \text{ if } &  \Upsilon_{\lambda,R}(e^{R-1}) > 1
                       \end{array}
  \right.,
\end{equation}
where $y$ is a solution of the equation
\begin{equation}\label{eq3:Ups=1}
\Upsilon_{\lambda,R}(y)=1\text{ under the condition }y>e^{R-1}.
\end{equation}

  \item
\begin{equation}\label{eq3:constantclamb}
C_{\lambda,R}\triangleq\lambda+\frac{R\cdot R!}{\beta_{\lambda,R}^{R-1}(Q_{\lambda,R})}\cdot\frac{2Q_{\lambda,R}}{Q_{\lambda,R}-1}.
\end{equation}
\end{itemize}
\end{notation}

Lemma \ref{lem3:Dmin} follows from Lemma
\ref{lem7:Dmin}. Theorem \ref{th3:infin_fam} summarizes the results of Sections \ref{sec:constrA}--\ref{sec:qmconc}.

\begin{lemma}\label{lem3:Dmin}
 We have
\begin{equation}\label{eq3:Dmin(N)<}
D^{\min}_R=\min_\lambda D_{\lambda,R}<\left\{\begin{array}{ccl}
                    1.651R & \text{ if } & R\ge3 \\
                    0.961R & \text{ if } & R\ge7 \\
                    0.498R & \text{ if }  & R\ge36\\
                    0.4R & \text{ if }  & R\ge178
                  \end{array}
\right..
\end{equation}
\end{lemma}
\begin{theorem}\label{th3:infin_fam}
 Let $R\ge3$ be fixed. Let the values used here correspond to Notation~\ref{not1}. For the length function $\ell_q(tR+1,R)$ and the smallest size $s_q(tR,R-1)$ of an $(R-1)$-saturating set in the
projective space $\PG(tR,q)$ we have the following upper bounds provided by infinite families of covering codes with growing codimension $r=tR+1$, $t\ge1$:
\begin{description}
  \item[(i)] \emph{(Upper bounds by a decreasing function)}

  If $q> Q_{\lambda,R}$, then $\Omega_{\lambda,R}(q)$ is a decreasing function of $q$ and $\Omega_{\lambda,R}(q)<C_{\lambda,R}$. Moreover,
\begin{equation}\label{eq3:boundDecrFun}
\begin{split}
& \ell_q(r,R)=s_q(r-1,R-1)<\Omega_{\lambda,R}(q)\cdot q^{(r-R)/R}\cdot\sqrt[R]{\ln q}+2Rq^{t-1}+R\theta_{t-1,q}\\
&<\left(\Omega_{\lambda,R}(q)+R\frac{2+q/(q-1)}{\sqrt[R]{q\ln q}}\right)q^{(r-R)/R}\cdot\sqrt[R]{\ln q},~r=tR+1,~t\ge1,~q> Q_{\lambda,R},
\end{split}
\end{equation}
where for $t\ge2$ the bound holds if $C_{\lambda,R}\sqrt[R]{q\ln q}+2R\le q+1$.

  \item[(ii)]  \emph{(Upper bounds by constants)}

Let $Q_0> Q_{\lambda,R}$ be a constant independent of $q$. Then $\Omega_{\lambda,R}(Q_0)$ is also a constant independent of $q$ such that $C_{\lambda,R}>\Omega_{\lambda,R}(Q_0)>D_{\lambda,R}$.
We have
\begin{equation}\label{eq3:boundConst}
\begin{split}
& \ell_q(r,R)=s_q(r-1,R-1)<cq^{(r-R)/R}\cdot\sqrt[R]{\ln q}+2Rq^{t-1}+R\theta_{t-1,q}\\
&<\left(c+R\frac{2+q_0/(q_0-1)}{\sqrt[R]{q_0\ln q_0}}\right)q^{(r-R)/R}\cdot\sqrt[R]{\ln q},~r=tR+1,~t\ge1,\\
&c\in\{C_{\lambda,R},\,\Omega_{\lambda,R}(Q_0)\},~q_0=\left\{\begin{array}{lcl}
                                                               Q_{\lambda,R} & \text{if} & c=C_{\lambda,R} \\
                                                               Q_0 & \text{if} & c=\Omega_{\lambda,R}(Q_0)
                                                             \end{array}
\right.,~q> q_0,
\end{split}
\end{equation}
where for $t\ge2$ the bound holds if $c\sqrt[R]{q\ln q}+2R\le q+1$.

  \item[(iii)] \emph{(Asymptotic upper bounds)}

   Let $q> Q_{\lambda,R}$ be large enough. Then the bounds \eqref{eq3:boundConstAsym} and  \eqref{eq3:boundConstAsym2} hold.
\begin{equation}\label{eq3:boundConstAsym}
\begin{split}
&\ell_q(r,R)=s_q(r-1,R-1)<cq^{(r-R)/R}\cdot\sqrt[R]{\ln q}+2Rq^{t-1}+R\theta_{t-1,q}\\
&<\left(c+R\frac{2+q/(q-1)}{\sqrt[R]{q\ln q}}\right)q^{(r-R)/R}\cdot\sqrt[R]{\ln q},~r=tR+1,~t\ge1,~c\in\{D^{\min}_R, D_{\lambda,R}\},
\end{split}
\end{equation}
where for $t\ge2$ the bounds hold if $c\sqrt[R]{q\ln q}+2R\le q+1$.
\begin{equation}\label{eq3:boundConstAsym2}
\ell_q(r,R)=s_q(r-1,R-1)<3.43Rq^{(r-R)/R}\cdot\sqrt[R]{\ln q},~r=tR+1,~t\ge1,
\end{equation}
where for $t\ge2$ the bounds hold if $D^{\min}_R\sqrt[R]{q\ln q}+2R\le q+1$.
\end{description}
\end{theorem}
Note that Theorem \ref{th6:t=1} in Section \ref{sec:boundsR_R-1} is a version of Theorem \ref{th3:infin_fam} for $t=1$.

Also, for $r=tR+1$ we have $t-1=(r-R-1)/R$ and $q^{(r-R)/R}=q^{t-1}\sqrt[R]{q}$.

\section{New Construction A of $(R-1)$-saturating sets in $\PG(R,q)$, $R\ge3$}\label{sec:constrA}

In this section, for any $q$ and $R\ge3$, we describe a new construction of $(R-1)$-saturating sets in $\PG(R,q)$. The points of such an $n$-set (in homogeneous coordinates), treated as columns, form a parity check matrix of an $[n,n-(R+1)]_qR$ code. In future, this code will be used as a starting one for lift-constructions obtaining infinite families of covering codes with growing codimension $r=tR+1$, $t\ge1$, see Section~\ref{sec:qmconc}.

\subsection{An iterative process}\label{subsec:iterproc}

We say that a point $P$ of $\PG(R,q)$ is \emph{$\rho$-covered} by a point set $\K\subset\PG(R,q)$ if $P$ lies in a $\overline{\rho}$-subspace generated by $\overline{\rho}+1$ points of $\K$ in general positions where $\overline{\rho}\le \rho$, see Definition~\ref{def1_usual satur}.  In this case,  the set~$\K$ \emph{$\rho$-covers} the point $P$.
If $\rho$ is clear by the context, one can say simply ``covered" and ``covers" (resp. ``uncovered'' and ``does not cover'').

Assume that in $\PG(R,q)$, an $(R-1)$-saturating set
 is constructed in a step-by-step iterative process adding $R$ new points to the set in
every step.

Let $A_u$ be a point of $\PG(R,q)$, $u=1,\ldots,\theta_{R,q}$.
Let $L>R$ be an integer.
Let
\begin{equation}\label{eq4:K0}
 \Pc_0=\{A_1,\ldots,A_{L}\}\subset\PG(R,q),~L>R,
\end{equation}
be a \emph{starting} $L$-set  such that any $R$ of its points are in general position.

In $\PG(R,q)$, an arc is a set of points no $R + 1$ of which belong to the same hyperplane.
So, any $R + 1$ points of an arc are in general position. For example, a normal rational curve is projectively equivalent to the arc
$\{(1,t,t^2, \ldots,t^R):t\in \F_q\}\cup \{(0,\ldots,0, 1)\}$.
We can take any $L$ points of any arc as the starting $L$-set.

Let $w\ge0$ be an integer. Let $\K_w$ be the current $(L+wR)$-set obtained after the
$w$-th step of the process; we put $\K_0=\Pc_0$, see \eqref{eq4:K0}.  Denote by
\begin{equation*}
\Pc_{w+1}=\{A_{L+wR+1},A_{L+wR+2},\ldots,A_{L+wR+R}\} \subset\PG(R,q),~w\ge0,
\end{equation*}
an $R$-set of points that are added to $\K_{w}$ on the $(w+1)$-st step to obtain $\K_{w+1}$. So,
\begin{equation}\label{eq4:Kw}
  \K_w=\Pc_0\cup\Pc_1\cup\ldots\cup\Pc_w\subset\PG(R,q),~\#\K_w=L+wR,~w\ge0.
\end{equation}

Let $\U_w$ be the subset of $\PG(R,q)\setminus\K_w$ consisting of the points that are \emph{not $(R-1)$-covered} by~$\K_w$.

The set $\Pc_{w+1}$ is constructed as follows.

Let $\Pi_w\subset\PG(R,q)$ be a hyperplane skew to $\K_w$. In $\PG(R,q)$, a blocking set regarding hyperplanes contains $\ge\theta_{1,q}$
points \cite{BoseBurt}. Therefore the saturating set with the size proved in this paper cannot be a blocking set. So, the needed $\Pi_w$ exists.

We put $\Pc_{w+1}\subset\Pi_w$. In the first, we choose a \emph{``leading point''} $\A_{w+1}\in\Pi_w$ and put  $A_{L+wR+1}=\A_{w+1}$ (the choice of the leading point is considered below). Then we take the points $A_{L+wR+2},\ldots,A_{L+wR+R}$ of $\Pi_w$ such that all the points of
$\Pc_{w+1}$ are in general position. Thus, $\Pc_{w+1}$ covers all points of~$\Pi_w$.

The iterative process is as follows:
\begin{itemize}
  \item We assign the starting set $\Pc_0$ in accordance to \eqref{eq4:K0} and put $w=0$, $\K_0=\Pc_0$.
  \item In every $(w+1)$-th step, we should do the following actions:

-- choose the leading point $\A_{w+1}$;

-- construct the $R$-set $\Pc_{w+1}$;

-- form the new current set $\K_{w+1}=\K_w\cup\Pc_{w+1}$;

-- count (or make an estimate of) the value $\#\U_{w+1}$.

  \item The process ends when $ \#\U_{w+1}\le R$. Finally, in the last $(w+1)$-step, we add to $\K_w$ at most $R$ uncovered points to obtain an $(R-1)$-saturating set.
\end{itemize}

\subsection{The choice of the leading point}\label{subsec:leadpoint}

 Let $\Delta_{w+1}(\Pc_{w+1})$ be the number of new covered points in $\U_w$ after adding $\Pc_{w+1}$ to $\K_w$;
 \begin{equation}\label{eq4:Delta_def}
   \Delta_{w+1}(\Pc_{w+1})=\#\U_w-\#\U_{w+1}.
 \end{equation}
We denote $\delta_{w+1}(\A_{w+1})$ the number of new covered points in $\U_w\setminus\Pi_w$ after adding the leading point $\A_{w+1}=A_{L+wR+1}$ to $\K_w$.
We have
\begin{equation}
  \Delta_{w+1}(\Pc_{w+1})\ge\delta_{w+1}(\A_{w+1})+\#(\U_w\cap\Pi_w)\ge\delta_{w+1}(\A_{w+1}),\label{eq4:Delta>=delta}
\end{equation}
where the first sign ``$\ge$'' is associated with the fact that the inclusion of the points $A_{L+wR+2},\ldots,A_{L+wR+R}$ can add new covered points outside $\Pi_w$.

 Let $\mathbb{S}_w$ be the sum of the number of new covered points in $\U_w\setminus\Pi_w$ over all points $P$ of $\Pi_w$, i.e.
 \begin{equation}\label{eq4:Sw}
   \mathbb{S}_w=\sum_{P\in\Pi_w}\delta_{w+1}(P).
 \end{equation}
The average value $ \delta_{w+1}^\text{aver}$ of $\delta_{w+1}(P)$ over all points of $\Pi_w$ is
 \begin{equation}\label{eq4:aver}
    \delta_{w+1}^\text{aver}=\frac{\sum\limits_{P\in\Pi_w}\delta_{w+1}(P)}{\#\Pi_w}=
    \frac{\mathbb{S}_w}{\theta_{R-1,q}}.
 \end{equation}
 Obviously, there exists a point $\A_{w+1}\in\Pi_w$ such that
 \begin{equation}\label{eq4:delta>aver}
 \delta_{w+1}(\A_{w+1})\ge\delta_{w+1}^\text{aver}.
 \end{equation}

The point $\A_{w+1}\in\Pi_w$ providing \eqref{eq4:delta>aver} should be chosen as the leading one.

\subsection{Estimates of the average number $ \delta_{w+1}^\text{aver}$ of new covered points}

To make the estimates, we introduce and consider a number of subspaces.

We denote by $\dim(H)$ the dimension of a subspace $H$.

 \emph{We fix a point $B\in\U_w\setminus\Pi_w$.} So, $B\notin\Pi_w$.

We consider $\binom{L}{R-1}$ \emph{distinct} $(R-1)$-subsets consisting of \emph{distinct points} of\linebreak $\K_{0}=\Pc_0$. We denote such a subset by
$\D_j$ with
 \begin{equation*}
   \D_j\subset\K_{0},~\#\D_j=R-1,~j=1,\ldots,\binom{L}{R-1},~\D_u\ne\D_v \text{ if }u\ne v.
 \end{equation*}
  By the assumptions, all the points of $\D_j$ are in general position.
  Also, all the points of the $R$-set $\D_j\cup\{B\}$  are in general position, otherwise $B$ would be covered by $\K_{0}$. Thus, the points of $\D_j\cup\{B\}$ uniquely define a hyperplane, say $\Sigma^{(R-1)}_{j,B}$, such that
  \begin{equation*}
   \Sigma^{(R-1)}_{j,B}=\langle\D_j\cup\{B\}\rangle\subset\PG(R,q),~\dim(\Sigma^{(R-1)}_{j,B})=R-1,~\#\Sigma^{(R-1)}_{j,B}=\theta_{R-1,q}.
  \end{equation*}

We have $\Sigma^{(R-1)}_{j,B}\ne\Pi_w$, as $B\notin\Pi_w$. Thus,
  $\Sigma^{(R-1)}_{j,B}$ and $\Pi_w$ intersect. The intersection is an $(R-2)$-subspace, say $\Gamma^{(R-2)}_{j,B}$, such that
  \begin{equation*}
    \Gamma^{(R-2)}_{j,B}=\Sigma^{(R-1)}_{j,B}\cap\Pi_w,~\dim(\Gamma^{(R-2)}_{j,B})=R-2,~\#\Gamma^{(R-2)}_{j,B}=\theta_{R-2,q}.
  \end{equation*}
 Let $ \V^{(R-2)}_j$ be the $(R-2)$-subspace generated by the  points of~$\D_j$, i.e.
\begin{equation*}
 \V^{(R-2)}_j=\langle\D_j\rangle\subset\Sigma^{(R-1)}_{j,B},~\dim(\V^{(R-2)}_j)=R-2,~\#\V^{(R-2)}_j=\theta_{R-2,q}.
\end{equation*}
As the $(R-2)$-subspaces $\V^{(R-2)}_j$ and $\Gamma^{(R-2)}_{j,B}$ lie in the same hyperplane $\Sigma^{(R-1)}_{j,B}$, they meet in some $(R-3)$-subspace, say $\T^{(R-3)}_{j,B}$, such that
\begin{equation*}
  \T^{(R-3)}_{j,B}=\V^{(R-2)}_j\cap\Gamma^{(R-2)}_{j,B},~\dim(\T^{(R-3)}_{j,B})=R-3,~\T^{(R-3)}_{j,B}=\theta_{R-3,q}.
\end{equation*}
The points of $\T^{(R-3)}_{j,B}$ are not in general position with the points of $\D_j$. We denote
 \begin{equation}\label{eq5:wideGamma}
 \widehat{\Gamma}^{(R-2)}_{j,B}=\Gamma^{(R-2)}_{j,B}\setminus\T^{(R-3)}_{j,B}.
 \end{equation}
 Every point of $\widehat{\Gamma}^{(R-2)}_{j,B}$ is in general position with the points of $\D_j$; also,
 \begin{equation*}
   \#\widehat{\Gamma}^{(R-2)}_{j,B}=\theta_{R-2,q}-\theta_{R-3,q}=q^{R-2}.
 \end{equation*}
 By the construction, the $q^{R-2}$-set $\widehat{\Gamma}^{(R-2)}_{j,B}$ is the affine point set of the $(R-2)$-subspace $\Gamma^{(R-2)}_{j,B}$.

   Thus, the hyperplane
  $\Sigma^{(R-1)}_{j,B}=\langle\D_j\cup\{B\}\rangle$ is generated  $q^{R-2}$ times when we  add in sequence all the points of $\Pi_w$ to $\K_w$ for the calculation of $\mathbb{S}_w$, see \eqref{eq4:Sw}.

The same holds for all $\binom{L}{R-1}$ sets $\D_j$. Moreover, consider the sets $\D_u$ and $\D_v$ with $u\ne v$. We have $\D_u\ne\D_v$. The points of $\D_u\cup\{B\}$ (resp. $\D_v\cup\{B\}$) define a hyperplane $\Sigma^{(R-1)}_{u,B}$ (resp. $\Sigma^{(R-1)}_{v,B}$). No points of $\D_v\setminus(\D_u\cap\D_v)$ lie in $\Sigma^{(R-1)}_{u,B}$, otherwise $B$ would be $(R-1)$-covered by~$\K_{0}$. So, the hyperplanes
 $\Sigma^{(R-1)}_{u,B}$ and $\Sigma^{(R-1)}_{v,B}$ are distinct.  If the corresponding $(R-2)$-subspaces $\Gamma^{(R-2)}_{u,B}=\Sigma^{(R-1)}_{u,B}\cap\Pi_w$ and $\Gamma^{(R-2)}_{v,B}=\Sigma^{(R-1)}_{v,B}\cap\Pi_w$ coincide with each other then $\Sigma^{(R-1)}_{u,B}$ and $\Sigma^{(R-1)}_{v,B}$ have no common points outside $\Pi_w$, contradiction as $B\notin\Pi_w$. Thus, $\Gamma^{(R-2)}_{u,B}\ne\Gamma^{(R-2)}_{v,B}$.

So, we have proved that in $\Pi_w$ we have $\binom{L}{R-1}$ distinct $(R-2)$-subspaces $\Gamma^{(R-2)}_{j,B}$ in every of which
the $q^{R-2}$-set $\widehat{\Gamma}^{(R-2)}_{j,B}$ of affine points gives rise to hyperplanes containing $B$.

 Thus, for the calculation of $\mathbb{S}_w$, the point $B$ will be counted $\#\G_{w,B}$
  times where
\begin{equation*}%\label{eq5:Gw}
  \G_{w,B}=\bigcup\limits_{j=1}^{\binom{L}{R-1}}\widehat{\Gamma}^{(R-2)}_{j,B}.
\end{equation*}
 The same holds for all points of $\U_w\setminus\Pi_w$. Therefore,
 \begin{equation}\label{eq5:sum}
\mathbb{S}_w=\sum_{P\in\Pi_w}\delta_{w+1}(P)\ge \sum_{B\in\U_w\setminus\Pi_w}\#\G_{w,B}
 \end{equation}
 where the sign ``$\ge$''  is needed as, for $w\ge1$, there are $(R-1)$-sets, say $\widetilde{D}_j$, $j>1$, consisting of points in general position, with $\widetilde{D}_j\subset\K_w$ and $\widetilde{D}_j\not\subset\K_0$. For example, every set $\Pc_w$, $w\ge1$, contains $R$ sets $\widetilde{D}_j$. Such sets together with uncovered points of $\U_w\setminus\Pi_w$ generate hyperplanes (similar to $\Sigma^{(R-1)}_{j,B}$) increasing $\mathbb{S}_w$.

 By \eqref{eq4:aver}, \eqref{eq5:sum}, for the average value $ \delta_{w+1}^\text{aver}$ of $\delta_{w+1}(P)$  we have
  \begin{equation}\label{eq5:aver2}
\delta_{w+1}^\text{aver}= \frac{\sum\limits_{P\in\Pi_w}\delta_{w+1}(P)}{\theta_{R-1,q}}\ge
 \frac{\sum\limits_{B\in\U_w\setminus\Pi_w}\#\G_{w,B}}{\theta_{R-1,q}}.
 \end{equation}

  The values of $\#\G_{w,B}$ can be distinct for distinct points $B$. Also, in principle, $\#\G_{w,B}$ can depend on $w$. We denote
\begin{equation}\label{eq5:Gmin_def}
\#\G^{\min}= \min_{B\in\U_W\setminus\Pi_w,\,W=1,\ldots,w}\#\G_{W,B}.
\end{equation}
Below, in Lemma \ref{lem4}, for the estimates of $\#\G^{\min}$, we use only the set $\K_0=\Pc_0$. Therefore, really, our estimates of $\#\G^{\min}$ do not depend on $w$.
By \eqref{eq5:aver2}, \eqref{eq5:Gmin_def},
\begin{equation}\label{eq5:aver3}
\delta_{w+1}^\text{aver}\ge\frac{\#\G^{\min}\cdot\#\U_w\setminus\Pi_w}{\theta_{R-1,q}}.
 \end{equation}

\begin{lemma}\label{lem4}
Let  $\binom{L}{R-1}-1\le q$. The following holds:
 \begin{equation}\label{eq5:min>=a}
\#\G^{\min}\ge q^{R-3}\binom{L}{R-1}\left(q+\frac{1}{2}-\frac{1}{2}\binom{L}{R-1}\right).
   \end{equation}
\end{lemma}

\begin{proof}
For some $n$, we consider $n$ of the $q^{R-2}$-sets $\widehat{\Gamma}^{(R-2)}_{j,B}$ of \eqref{eq5:wideGamma}. All the sets are distinct; in fact, if
 $\widehat{\Gamma}^{(R-2)}_{u,B} = \widehat{\Gamma}^{(R-2)}_{v,B} , u\ne v$, then $\widehat{\Gamma}^{(R-2)}_{u,B} \subset  \Gamma^{(R-2)}_{u,B}\cap\Gamma^{(R-2)}_{v,B}$ that implies
 $q^{R-2} =  \#\widehat{\Gamma}^{(R-2)}_{u,B} < \#(\Gamma^{(R-2)}_{u,B}\cap\Gamma^{(R-2)}_{v,B})=\theta_{R-3,q} $, contradiction.

 As $\widehat{\Gamma}^{(R-2)}_{u,B}$ and $\widehat{\Gamma}^{(R-2)}_{v,B}$ are the affine point sets of the distinct $(R-2)$-spaces, they have at most $q^{R-3}$ points in common, i.e. $\#(\widehat{\Gamma}^{(R-2)}_{u,B} \cap \widehat{\Gamma}^{(R-2)}_{v,B}) \leq q^{R-3}$.

Assume that $\#(\widehat{\Gamma}^{(R-2)}_{u,B} \cap \widehat{\Gamma}^{(R-2)}_{v,B}) = q^{R-3}$, for all pairs $(u,v)$, and that, in every set $\widehat{\Gamma}^{(R-2)}_{j,B}$, all the intersection points are distinct; it is the worst case  for $\#\G_{w,B}$.

In every set $\widehat{\Gamma}^{(R-2)}_{j,B}$, the number of the affine point sets  intersecting it is $n-1$ and the number of the intersection points  is $(n-1)q^{R-3}$. As $q^{R-2}-(n-1)q^{R-3}$ must be $\ge0$, the considered case is possible if $n-1\le q$.

In all $n$ sets $\widehat{\Gamma}^{(R-2)}_{j,B}$, the total
number of the intersection points  is $n(n-1)q^{R-3}$. The total number $\#\G(n)$ of distinct points in the union $\G(n)=\bigcup_{j=1}^n\widehat{\Gamma}^{(R-2)}_{j,B}$ is $\#\G(n)=nq^{R-2}-\frac{1}{2}n(n-1)q^{R-3}$ where $q^{R-2}=\#\widehat{\Gamma}^{(R-2)}_{j,B}$ and we need the factor~$\frac{1}{2}$ in order to calculate the meeting points exactly one time.

Finally, we put $n=\binom{L}{R-1}$.
\end{proof}

\begin{remark}
The condition $\binom{L}{R-1}-1\le q$ is used below and gives rise that our estimates work for $q>Q_{\lambda,R}$, see \eqref{eq3:Qlx} and Theorems \ref{th3:infin_fam} and \ref{th6:t=1}. In principle, we could slightly change the proof of Lemma~\ref{lem4} and put either $n=\binom{L}{R-1}$ if $\binom{L}{R-1}-1\le q$ or $n=q+1$ if $\binom{L}{R-1}-1> q$. This gives the estimate
\begin{equation}\label{eq4:newEst}
 \#\G^{\min}\ge\left\{
 \begin{array}{c}
q^{R-3}\binom{L}{R-1}\left(q+\frac{1}{2}-\frac{1}{2}\binom{L}{R-1}\right)\text{ if }  \binom{L}{R-1}-1\le q\smallskip\\
                      \frac{1}{2}(q^{R-1}+q^{R-2})\text{ if }\binom{L}{R-1}-1> q
\end{array}
  \right..
\end{equation}
On the base of \eqref{eq4:newEst}, upper bounds for $q<Q_{\lambda,R}$ could be obtained. We do not it for the sake of simplicity. We hope investigate the case $q<Q_{\lambda,R}$ in future works.
\end{remark}

\section{Estimates of sizes of the saturating sets obtained by  Construction~A}\label{sec:est_size}

\begin{lemma}
  For the number $\#\U_{w+1}$ of uncovered points after the $(w+1)$-st step of the iterative process, we have
\begin{equation}\label{eq6:Rwp1}
 \#\U_{w+1}\le q^R\left(1- \frac{\#\G^{\min}}{\theta_{R-1,q}}\right)^{w+1}.
\end{equation}
\end{lemma}

\begin{proof}
By \eqref{eq4:Delta_def}, \eqref{eq4:Delta>=delta}, \eqref{eq4:delta>aver}, \eqref{eq5:aver2}, \eqref{eq5:aver3}, we have
 \begin{equation*}
 \begin{split}
\Delta_{w+1}(\Pc_{w+1})=\#\U_w-\#\U_{w+1}&=\#\U_w\setminus\Pi_w+\#(\U_w\cap\Pi_w)-\#\U_{w+1}\\
 &\ge\frac{\#\G^{\min}\cdot\#\U_w\setminus\Pi_w}{\theta_{R-1,q}}+\#(\U_w\cap\Pi_w),
 \end{split}
 \end{equation*}
 where $\#\G^{\min}\cdot\#\U_w\setminus\Pi_w$ is the lower bound of $\sum_{B\in\U_w\setminus\Pi_w}\#\G_{w,B}$, see \eqref{eq5:aver2}. Therefore, $(\#\G^{\min}\cdot\#\U_w\setminus\Pi_w)/\theta_{R-1,q}$ is the lower bound of the number of the new covered points in $\U_w\setminus\Pi_w$. It follows that $\#\G^{\min}/\theta_{R-1,q}\le1$, as the new covered points in the set $\U_w\setminus\Pi_w$ are a subset of it that implies $(\#\G^{\min}\cdot\#\U_w\setminus\Pi_w)/\theta_{R-1,q}\le\#\U_w\setminus\Pi_w$. The summand $\#(\U_w\cap\Pi_w)$ takes into account that $\Pc_{w+1}$ covers all points of $\Pi_w$, see Subsection \ref{subsec:iterproc}.

 As $\#\G^{\min}/\theta_{R-1,q}\le1$ and $\#\U_w=\#\U_w\setminus\Pi_w+\#(\U_w\cap\Pi_w)$, we obtain
 \begin{equation*}
\Delta_{w+1}(\Pc_{w+1}) \ge\frac{\#\G^{\min}\cdot\#\U_w}{\theta_{R-1,q}};
 \end{equation*}
 \begin{equation} \label{eq6:Rwp1_Rw}
    \#\U_{w+1}\le\#\U_w-\frac{\#\G^{\min}\cdot\#\U_w}{\theta_{R-1,q}}=
    \#\U_w\left(1- \frac{\#\G^{\min}}{\theta_{R-1,q}}\right).
 \end{equation}

As any $R$ points of $\K_{0}$ are in general position, we have
\begin{equation*}
\#\U_{0}\le \theta_{R,q}-\theta_{R-1,q}=q^R.
\end{equation*}
Starting from $\#\U_{0}$ and iteratively applying \eqref{eq6:Rwp1_Rw}, we obtain the assertion.
\end{proof}

By Notation \ref{not1}, $\lambda$ is a positive constant that does not depend on $q$.  Let
\begin{equation}\label{eq6:w0}
L=\left\lfloor \lambda \sqrt[R]{q\ln q}\right\rfloor
\end{equation}
that implies
\begin{equation}\label{eq6:w0b}
 \lambda\sqrt[R]{q\ln q}-1< L\le \lambda\sqrt[R]{q\ln q}.
 \end{equation}
  From \eqref{eq3:beta0_def} and \eqref{eq6:w0b} we have
  \begin{equation}\label{eq6:L-Rp1>=}
  L-R+1\le\beta_{\lambda,R}(q) \sqrt[R]{q\ln q}<  L-R+2.
  \end{equation}

We denote
 \begin{equation}\label{eq6:widePhi_def}
  \Phi_{\lambda,R}^*(q)=\frac{2q}{2q-1-\binom{L}{R-1}}\, .
 \end{equation}

 \begin{lemma}
 Let  $\binom{L}{R-1}-1\le q$. The following holds:
\begin{equation}\label{eq6:estim0}
 \left(1-  \frac{\#\G^{\min}}{\theta_{R-1,q}}\right)^{w+1}<\exp\left(-\frac{(w+1)\binom{L}{R-1}}
{q\Phi_{\lambda,R}^*(q)}\right).
\end{equation}
 \end{lemma}

 \begin{proof}
 By the inequality $1-x\le\exp(-x)$ and by \eqref{eq5:min>=a}, if $\binom{L}{R-1}-1\le q$ we have
\begin{equation*}
\begin{split}
 & \left(1-  \frac{\#\G^{\min}}{\theta_{R-1,q}}\right)^{w+1}<\exp\left(-\frac{(w+1)\cdot\#\G^{\min}}{\theta_{R-1,q}}\right)\\
& <\exp\left(-(w+1)q^{R-3}\binom{L}{R-1}\left(q+\frac{1}{2}-\frac{1}{2}\binom{L}{R-1}\right)
\frac{q-1}{q^{R}-1}\right)\\
&<\exp\left(-(w+1)q^{R-3}\binom{L}{R-1}\left(2q+1-\binom{L}{R-1}\right)
\frac{q-1}{2q^{R}}\right) \\
&<\exp\left(-(w+1)\binom{L}{R-1}\left(2q^2-q-q\binom{L}{R-1}+\binom{L}{R-1}-1\right)
\frac{1}{2q^{3}}\right) \\
&<\exp\left(-(w+1)\binom{L}{R-1}\left(2q^2-q-q\binom{L}{R-1}\right)
\frac{1}{2q^{3}}\right)
\end{split}
\end{equation*}
where the last transformation uses that, by \eqref{eq4:K0}, $L>R$ and $\binom{L}{R-1}-1>0$.
Therefore, removing $\binom{L}{R-1}-1$ we obtain the inequality ``$<\exp\left(-\ldots\right)$''.
 \end{proof}

\begin{proposition}\label{prop6}
Let $\binom{L}{R-1}-1\le q.$  Then the  value
  \begin{equation}\label{eq6:w>=}
   w\ge \frac{R!}{\beta_{\lambda,R}^{R-1}(q)}\Phi_{\lambda,R}^*(q)\sqrt[R]{q\ln q}-1.
  \end{equation}
  satisfies the inequality $\#\U_{w+1}\le R$.
\end{proposition}

\begin{proof}
  By \eqref{eq6:Rwp1}, \eqref{eq6:estim0}, to prove $\#\U_{w+1}\le R$ it is sufficient to find $w$ such that
\begin{equation*}
  \exp\left(-\frac{(w+1)\binom{L}{R-1}}{q\Phi_{\lambda,R}^*(q)}\right)\le\frac{R}{q^R}.
\end{equation*}
Using \eqref{eq6:L-Rp1>=}, we obtain
\begin{equation*}
\begin{split}
\exp\left(-\frac{(w+1)\binom{L}{R-1}}{q\Phi_{\lambda,R}^*(q)}\right)& <\exp\left(-\frac{(w+1)(L-R+2)^{R-1}}{(R-1)!\cdot q\Phi_{\lambda,R}^*(q)}\right)\\
 &<\exp\left(-\frac{(w+1)\left(\beta_{\lambda,R}(q)\sqrt[R]{q\ln q}\,\right)^{R-1}}{(R-1)!\cdot q\Phi_{\lambda,R}^*(q)}\right) \le\frac{R}{q^R}.
\end{split}
\end{equation*}
Taking the logarithm of both the parts of the last inequality, we have
\begin{equation*}
 \frac{(w+1)\left(\beta_{\lambda,R}(q)\sqrt[R]{q\ln q}\,\right)^{R-1}}{(R-1)!\cdot q\Phi_{\lambda,R}^*(q)}\ge R\ln q-\ln R;
\end{equation*}
\begin{equation*}
 w\ge(R\ln q-\ln R )\frac{(R-1)!\cdot q\Phi_{\lambda,R}^*(q)}{\left(\beta_{\lambda,R}(q)\sqrt[R]{q\ln q}\,\right)^{R-1}}-1.
\end{equation*}
If $a\ge R\ln q$ then $a\ge R\ln q-\ln R$. Therefore we may use
the inequality
\begin{equation*}
 w\ge R\ln q\frac{(R-1)!\cdot q\Phi_{\lambda,R}^*(q)}{\left(\beta_{\lambda,R}(q)\sqrt[R]{q\ln q}\,\right)^{R-1}}-1
\end{equation*}
which slightly worsens our estimates but simplifies the transformations. Finally,
\begin{equation*}
 w\ge \frac{R!\cdot\Phi_{\lambda,R}^*(q)q\ln q}{\beta_{\lambda,R}^{R-1}(q)\sqrt[R]{q^{R-1}\ln^{R-1} q}}-1=\frac{R!}{\beta_{\lambda,R}^{R-1}(q)}\Phi_{\lambda,R}^*(q)\sqrt[R]{q\ln q}-1. \qedhere
\end{equation*}
\end{proof}

We denote, see \eqref{eq3:Omega_def} and \eqref{eq6:widePhi_def},
\begin{equation}\label{eq6:wideOmega_def}
 \Omega_{\lambda,R}^*(q)=\lambda+\frac{R\cdot R!}{\beta_{\lambda,R}^{R-1}(q)}\Phi_{\lambda,R}^*(q).
\end{equation}

\begin{theorem}
  In  $\PG(R,q)$, for the size $s^\text{A}_{R-1,q}$ of the $(R-1)$-saturating set obtained by Construction A and for the smallest size $s_q(R,R-1)$ of an $(R-1)$-saturating set
the following upper bound holds:
\begin{equation}\label{eq6:sqRrho}
s_q(R,R-1)\le s^\text{A}_{R-1,q}<\Omega_{\lambda,R}^*(q) \sqrt[R]{q\ln q}+2R\text{ if }\binom{L}{R-1}-1\le q.
\end{equation}
\end{theorem}

\begin{proof}
  By \eqref{eq4:Kw}, \eqref{eq6:w0}, \eqref{eq6:w0b}, and Proposition \ref{prop6} with \eqref{eq6:w>=}, Construction A obtains an $(R-1)$-saturating set of size
  \begin{equation*}
  \begin{split}
L+(w+1)R+R&=  L+\left(\left\lceil\frac{R!}{\beta_{\lambda,R}^{R-1}(q)}\Phi_{\lambda,R}^*(q)\sqrt[R]{q\ln q} \right\rceil\right)R+R\\
 &\le \lambda\sqrt[R]{q\ln q}+\left(\frac{R!}{\beta_{\lambda,R}^{R-1}(q)}\Phi_{\lambda,R}^*(q)\sqrt[R]{q\ln q}+1\right)R+R.
  \end{split}
  \end{equation*}
 Now the assertion follows due to \eqref{eq6:wideOmega_def}.
\end{proof}

\section{Upper bounds on the length function $\ell_q(R+1,R)$, $R\ge3$}\label{sec:boundsR_R-1}

 By \eqref{eq3:Upsilon_def}, \eqref{eq6:w0}, \eqref{eq6:w0b}, we have
   \begin{equation}\label{eq7:binomw0rho<=}
    \binom{L}{R-1}<\frac{L^{R-1}}{(R-1)!}\le\frac{\lambda^{R-1}}{(R-1)!}\sqrt[R]{q^{R-1}\ln^{R-1} q}=q\Upsilon_{\lambda,R}(q).
   \end{equation}

   \begin{lemma}\label{lem7:Upsilon<=1}
    The condition $\binom{L}{R-1}-1\le q$ holds if
    \begin{equation}
\Upsilon_{\lambda,R}(q)\le 1.
    \end{equation}
  \end{lemma}

  \begin{proof}
By \eqref{eq7:binomw0rho<=}, we have $\binom{L}{R-1}-1\le q$ if $\Upsilon_{\lambda,R}(q)\le(q+1)/q$. For simplicity of presentation we consider $\Upsilon_{\lambda,R}(q)\le 1$.
  \end{proof}

 \begin{lemma}\label{lem7:Ups_decreas}
 Let $\lambda$ and $R$ be fixed. Let $q>e^{R-1}$. Then $\Upsilon_{\lambda,R}(q)$ is a decreasing function of~$q$.
      \end{lemma}

  \begin{proof} The derivative
  \begin{equation*}
  \left(\frac{\ln^{R-1} q}{q}\right)'=\frac{(R-1)\ln^{R-2}q-\ln^{R-1}q}{q^2}
  \end{equation*}
  is negative when $\ln q>R-1$.
  \end{proof}

\begin{corollary}\label{cor7}
  We have
  \begin{equation}
    Q_{\lambda,R}>e^{R-1}.
  \end{equation}
\end{corollary}

\begin{proof}
  The assertion follows from \eqref{eq3:Qlx}, \eqref{eq3:Ups=1}, and Lemma \ref{lem7:Ups_decreas}.
\end{proof}

\begin{remark}
  Note that \eqref{eq3:Ups=1} is equivalent to the equation
\begin{equation*}
\ln^{R-1} y=y\left(\frac{(R-1)!}{\lambda^{R-1}}\right)^R~\text{ under the condition }y>e^{R-1}.
\end{equation*}
This equation is connected with Lambert $W$ function, see e.g. \cite{Lambert}.
\end{remark}

The following two lemmas are obvious.

\begin{lemma}\label{lem7:Ups_infin}
 Let $\lambda$ and $R$ be fixed. Then
    \begin{equation}\label{eq7:limUpsilon}
     \lim_{q\rightarrow\infty} \Upsilon_{\lambda,R}(q)=0.
    \end{equation}
      \end{lemma}

  \begin{lemma}\label{lem7:beta0}
    Let $\lambda$ and $R$ be fixed. Then $\beta_{\lambda,R}(q)$ of \eqref{eq3:beta0_def} is an increasing function of $q$ and
    \begin{equation}\label{eq7:limbeta}
     \lim_{q\rightarrow\infty} \beta_{\lambda,R}(q)=\lambda.
    \end{equation}
  \end{lemma}

We introduce a function $\Phi_{\lambda,R}(q)$ of $q$, cf. \eqref{eq6:widePhi_def}, \eqref{eq7:binomw0rho<=}, and \eqref{eq3:Omega_def},
\begin{equation}\label{eq7:Phi_def}
\Phi_{\lambda,R}(q) =
\frac{2}{2-\frac{1}{q}-\frac{\lambda^{R-1}}{(R-1)!}\sqrt[R]{\frac{\ln^{R-1} q}{q}}}=\frac{2}{2-\frac{1}{q}-\Upsilon_{\lambda,R}(q)}.
\end{equation}

\begin{lemma}\label{lem7:properties}
 Let the values used here be as in Notation \ref{not1} and in \eqref{eq7:Phi_def}. The following holds.
  \begin{description}
    \item[(i)] Let $\lambda$, $R$ be fixed. Let $q>e^{R-1}$. Then $\Phi_{\lambda,R}(q)$ and $\Omega_{\lambda,R}(q)$ are decreasing functions of~$q$.

    \item[(ii)] Let $\lambda$, $R$ be fixed. Let $q>e^{R-1}$. Then $\Phi_{\lambda,R}^*(q)$ and $\Omega_{\lambda,R}^*(q)$ are upper bounded by decreasing functions of $q$ such that
\begin{equation}\label{eq7:Phi<1}
\Phi_{\lambda,R}^*(q)<\Phi_{\lambda,R}(q);
    \end{equation}
\begin{equation} \label{eq7:Omega<}
\Omega_{\lambda,R}^*(q)<\Omega_{\lambda,R}(q).
    \end{equation}

    \item[(iii)]  If $q> Q_{\lambda,R}$ then
\begin{equation}\label{eq7:Phi<2}
      \Phi_{\lambda,R}(q)^*<\Phi_{\lambda,R}(q)<\Phi_{\lambda,R}(Q_{\lambda,R})\le\frac{2Q_{\lambda,R}}{Q_{\lambda,R}-1};
     \end{equation}
\begin{equation}\label{eq7:beta>}
\beta_{\lambda,R}(q)>\beta_{\lambda,R}(Q_{\lambda,R})=\lambda-\frac{R-1}{\sqrt[R]{Q_{\lambda,R}\ln Q_{\lambda,R}}};
     \end{equation}
\begin{equation}\label{eq7:Omega<2}
\Omega_{\lambda,R}(q)<C_{\lambda,R}.
     \end{equation}

    \item[(iv)]
    Let $\lambda$, $R$ be fixed. Then
\begin{equation}\label{eq7:LimOmega}
 \lim_{q\rightarrow\infty}\Omega_{\lambda,R}(q)=\lambda+\frac{R\cdot R!}{\lambda^{R-1}}=D_{\lambda,R}.
\end{equation}
  \end{description}
\end{lemma}

\begin{proof}
\begin{description}
  \item[(i)]
We consider $\Phi_{\lambda,R}(q)$. By Lemma \ref{lem7:Ups_decreas}, if $q>e^{R-1}$ then $\Upsilon_{\lambda,R}(q)$ is a decreasing function of $q$. It implies that the
 function $2/(2-\frac{1}{q}-\Upsilon_{\lambda,R}(q))$ is decreasing also.

 Now we consider $\Omega_{\lambda,R}(q)=\lambda+\frac{R\cdot R!}{\beta_{\lambda,R}^{R-1}(q)}\Phi_{\lambda,R}(q)$, see \eqref{eq3:Omega_def}, \eqref{eq7:Phi_def}. The assertion on $\Omega_{\lambda,R}(q)$ follows from Lemma \ref{lem7:beta0} and the first part of this case~(i).

  \item[(ii)] The assertion \eqref{eq7:Phi<1} follows from \eqref{eq6:widePhi_def}, \eqref{eq7:binomw0rho<=},  and \eqref{eq7:Phi_def}.

 For \eqref{eq7:Omega<} we use \eqref{eq3:beta0_def}, \eqref{eq3:Omega_def}, \eqref{eq6:wideOmega_def}, and \eqref{eq7:Phi<1}.

  \item[(iii)] If $q> Q_{\lambda,R}$ then $q> e^{R-1}$, see Corollary \ref{cor7}, and $\Phi_{\lambda,R}(q)$ is a decreasing function, see the case (i).
 This implies $\Phi_{\lambda,R}(q)<\Phi_{\lambda,R}(Q_{\lambda,R})$. Then we use~\eqref{eq7:Phi<1}. Moreover, by \eqref{eq7:Phi_def}, \eqref{eq3:Qlx}, and Lemma~\ref{lem7:Ups_decreas}, we have
\begin{equation*}
\Phi_{\lambda,R}(Q_{\lambda,R})=\frac{2}{2-\frac{1}{Q_{\lambda,R}}-\Upsilon_{\lambda,R}(Q_{\lambda,R})}<
\frac{2}{1-\frac{1}{Q_{\lambda,R}}}=\frac{2Q_{\lambda,R}}{Q_{\lambda,R}-1}.
 \end{equation*}

 By Lemma \ref{lem7:beta0}, $\beta_{\lambda,R}(q)$ is an increasing function of $q$ that implies \eqref{eq7:beta>}.

 By Corollary \ref{cor7} and the case (i) of this lemma, $\Omega_{\lambda,R}(q)$ is a decreasing function, if $q> Q_{\lambda,R}$, and we have $\Omega_{\lambda,R}(q)<\Omega_{\lambda,R}(Q_{\lambda,R})$. Then we use \eqref{eq3:constantclamb}, \eqref{eq7:Phi<2} and obtain \eqref{eq7:Omega<2}.
  \item[(iv)]
  We use \eqref{eq3:constantDlamb}, \eqref{eq3:Omega_def}, \eqref{eq7:limUpsilon}, and \eqref{eq7:limbeta}. \qedhere
\end{description}
\end{proof}

We denote
\begin{equation}\label{eq6:lambmin_def}
  \lambda_{\min}\triangleq\sqrt[R]{R(R-1)\cdot R!}.
\end{equation}
\begin{lemma}\label{lem7:Dmin}
Let $R\ge3$ be fixed.
\begin{description}
    \item[(i)]
    The minimum value $D^{\min}_R$ of $D_{\lambda,R}$ regarding $\lambda$ is as follows:
    \begin{equation}\label{eq6:Dmin}
      D^{\min}_R\triangleq\min_\lambda D_{\lambda,R}=D_{\lambda_{\min},R}=\frac{R}{R-1}\sqrt[R]{R(R-1)\cdot R!}.
    \end{equation}
  \item[(ii)]  $\frac{1}{R}D^{\min}_R$ is a decreasing function of $R$.
  \item[(iii)] The relation \eqref{eq3:Dmin(N)<} holds.

  \item[(iv)] Let $\lambda= \lambda_{\min}$ and $D_{\lambda,R}=D^{\min}_R$. Then the size of the starting set of \eqref{eq6:w0} is $L>R$ if $q$ satisfies any of the following lower bounds:
      \begin{equation}\label{eq6:lowbndq}
       q\ln q>\frac{R^R}{R!\cdot R(R-1)};~q>\frac{R^R}{R!}.
      \end{equation}
\end{description}
\end{lemma}

\begin{proof}
\begin{description}
    \item[(i)]
  We have $D'_{\lambda,R}=1-\lambda^{-R}(R-1)R\cdot R!$ where $D'_{\lambda,R}$ is the derivative of $D_{\lambda,R}$ as a function of $\lambda$. From $D'_{\lambda,R}=0$ follows $\lambda=\sqrt[R]{R(R-1)\cdot R!}=\lambda_{\min}$. It corresponds to the minimal $D_{\lambda,R}$ as the derivative is an increasing function of $\lambda$.
We substitute $\lambda_{\min}$ to \eqref{eq3:constantDlamb} that gives the last equality of \eqref{eq6:Dmin}.

  \item[(ii)]  We have $D^{\min}_R/R= \sqrt[R]{R^2 \cdot(R-1)!/(R-1)^{R-1}}$. The derivative  $(R^{2/R})'=2R^{2/R-2}(1-\ln R)<0$ as $R\ge3$; so, $\sqrt[R]{R^2}$ is a decreasing function of $R$. Also, $(R-1)!/(R-1)^{R-1}$ is a decreasing function of $R$ if $R\ge2$.
  \item[(iii)] For $R=3,7,36,178$, we directly calculate $D^{\min}_R$ by \eqref{eq3:constantDmin}. Then we use the case (ii) of this lemma.  See also Table \ref{tab1} below for the illustration.
  \item[(iv)]  The assertion follows from \eqref{eq6:w0} and \eqref{eq6:lambmin_def}. \qedhere
      \end{description}
\end{proof}

\begin{remark}
  It can be shown that
  \begin{align*}
   \lim_{R\rightarrow\infty} \frac{D^{\min}_R}{R}=\frac{1}{e}\thickapprox0,3679.
  \end{align*}
  So, the relation \eqref{eq3:Dmin(N)<} is convenient for estimates  of the new asymptotic bounds.
\end{remark}

Now we can prove Theorem \ref{th6:t=1} that is a version of Theorem \ref{th3:infin_fam} for $t=1$.

\begin{theorem}\label{th6:t=1}
 Let $R\ge3$ be fixed.  Let the values used here correspond to Notation~\ref{not1}. For the length function $\ell_q(R+1,R)$ and the smallest size $s_q(R,R-1)$ of an $(R-1)$-saturating set in the
projective space $\PG(R,q)$ the following upper bounds hold:
\begin{description}
  \item[(i)] \emph{(Upper bound by a decreasing function)}

  If $q>Q_{\lambda,R}$, then $\Omega_{\lambda,R}(q)$ is a decreasing function of $q$ and $\Omega_{\lambda,R}(q)<C_{\lambda,R}$. Moreover,
\begin{equation}\label{eq5:boundDecrFun}
\ell_q(R+1,R)=s_q(R,R-1)<\Omega_{\lambda,R}(q)\sqrt[R]{q\ln q}+2R,~q> Q_{\lambda,R}.
\end{equation}

  \item[(ii)]  \emph{(Upper bounds by constants)}

 Let $Q_0> Q_{\lambda,R}$ be a constant independent of $q$. Then $\Omega_{\lambda,R}(Q_0)$ is also a constant independent of $q$ such that $C_{\lambda,R}>\Omega_{\lambda,R}(Q_0)>D_{\lambda,R}$.
We have
\begin{equation}\label{eq6:boundConst}
\begin{split}
& \ell_q(R+1,R)=s_q(R,R-1)<c\sqrt[R]{q\ln q}+2R<\left(c+\frac{2R}{\sqrt[R]{q_0\ln q_0}}\right)\sqrt[R]{q\ln q},\\
&c\in\{C_{\lambda,R},\,\Omega_{\lambda,R}(Q_0)\},~q_0=\left\{\begin{array}{lcl}
                                                               Q_{\lambda,R} & \text{if} & c=C_{\lambda,R} \\
                                                               Q_0 & \text{if} & c=\Omega_{\lambda,R}(Q_0)
                                                             \end{array}
\right.,~q> q_0.
\end{split}
\end{equation}

  \item[(iii)] \emph{(Asymptotic upper bounds)}

   Let $q> Q_{\lambda,R}$ be large enough. Then the bounds \eqref{eq6:boundConstAsym} and \eqref{eq6:boundConstAsym2} hold.
\begin{equation}\label{eq6:boundConstAsym}
\ell_q(R+1,R)=s_q(R,R-1)<c\sqrt[R]{q\ln q}+2R,~c\in\{D^{\min}_R, D_{\lambda,R}\}.
\end{equation}
\begin{equation}\label{eq6:boundConstAsym2}
\begin{split}
&\ell_q(R+1,R)=s_q(R,R-1)<1.651R\sqrt[R]{q\ln q}+2R.
\end{split}
\end{equation}
\end{description}
\end{theorem}

\begin{proof}
\begin{description}
  \item[(i)]We use Lemma \ref{lem7:properties}(i) and \eqref{eq6:sqRrho},  \eqref{eq7:Omega<}.
  \item[(ii)] The assertion follows from \eqref{eq3:Upsilon_def}, \eqref{eq3:Omega_def}, \eqref{eq6:sqRrho}, \eqref{eq7:Omega<}, \eqref{eq7:Omega<2}, and Lemma \ref{lem7:properties}(i),(iv) which simplifies proving $\Omega_{\lambda,R}(Q_0)>D_{\lambda,R}$.
  \item[(iii)]  We use the case (i) of this theorem and  Lemmas \ref{lem7:properties}(iv) and \ref{lem7:Dmin}. \qedhere
\end{description}
\end{proof}

We call the value $C_{\lambda,R}+\frac{2R}{\sqrt[R]{Q_{\lambda,R}\ln Q_{\lambda,R}}}$ the \emph{basic constant} for $t=1$, see \eqref{eq6:boundConst}.
\begin{example}
In Table \ref{tab1}, examples of values connected with upper bounds of Theorem \ref{th6:t=1} and Theorem \ref{th3:infin_fam} for $t=1$ are given. For every $R$, the last value of $\lambda$ is $\lambda_{\min}$, see \eqref{eq6:lambmin_def} and Lemma \ref{lem7:Dmin}; it gives rise to $\min\limits_\lambda D_{\lambda,R}=D^{\min}_R$. In the table, the values of $R$, $Q_{\lambda,R}$, and $\lambda$, apart from $\lambda_{\min}$, are exact. The rest of them are approximate.

    \begin{table}[htbp]
    \centering
    \caption{Examples of values connected with upper bounds of Theorem \ref{th6:t=1} and Theorem \ref{th3:infin_fam} for $t=1$; $Q_0\in\{5\cdot10^4,15\cdot10^4\}$, $E=e^{R-1}$}\medskip
    \label{tab1}
    \begin{tabular}{@{}c|c|c|r|r|r|r|c@{}}
      \hline
      $R$ & $\lambda$ & $\Upsilon_{\lambda,R}(E)$ & $Q_{\lambda,R}$ & $C_{\lambda,R}$ &$\Omega_{\lambda,R}(Q_{0})$&$\Omega_{\lambda,R}(Q_0)$& $D_{\lambda,R}$\\
      $E$&&&&&$Q_0=$&$Q_0=$&\\
      &&&&&$5\cdot10^4$&$15\cdot10^4$&\\\hline
      3 &2.35&2.25&1007&9.50&6.43&6.17&5.61\\
      7.39&3& 3.67&7186&7.14&5.90&5.60&5\\
      &$\lambda_{\min}=$&4.44&14974&6.69&5.93&5.58&$4.953=D^{\min}_R$\\
      &3.302&&&&&&$=1.651R$\\\hline

      4&2.2&1.91&6826&25.9&18.49&16.42&11.22\\
      20.1&2.5&2.80&61724&16.5&&14.30&8.64\\
      &$\lambda_{\min}=$&12.55&118409572&6.89&&&$5.493=D^{\min}_R$\\
      &4.120&&&&&&$=1.373R$\\\hline

      5&2.3&1.59&21242&84.3&68.53&55.4&23.74\\
      54.6&2.5&2.22&283935&45.1&&&17.86\\
      &$\lambda_{\min}=$&28.72&&&&&$5.929=D^{\min}_R$\\
      &4.743&&&&&&$=1.186R$\\\hline

      6&2.5&1.35&37774&337&304.6&217.7&46.73\\
      148&$\lambda_{\min}=$&56.67&&&&&$6.333=D^{\min}_R$\\
      &5.277&&&&&&$=1.056R$\\\hline

      7&2.95&1.80&9125037&265&&&56.48\\
      403&$\lambda_{\min}=$&100.5&&&&&$6.726=D^{\min}_R$\\
      &5.765&&&&&&$=0.961R$\\\hline
    \end{tabular}
  \end{table}
\end{example}

\section{Bounds on the length function $\ell_q(4,3)$ and 2-saturating sets in $\PG(3,q)$}\label{sec:R=3}

For illustration, we compare the bounds obtained by computer search with the theoretical bounds of Theorems \ref{th6:t=1} and \ref{th3:infin_fam} for $R=3$, $t=1$.

 Complete arcs in  $\PG(3,q)$ are 2-saturating sets. In \cite{BDMP-arXivR3_2018,DMP-R=3Redun2019,DMP-ICCSA2020,DavOst-IEEE2001}, see also the references therein, small complete arcs in $\PG(3,q)$ for the region $13\le q\le7057$ are obtained by computer search using the so-called ``algorithms with the fixed order of points (FOP)'' and ``randomized greedy algorithms''. These algorithms are described in detail in \cite{BDMP-arXivR3_2018,DMP-R=3Redun2019}.

In this paper, we continue the computer search and obtain new small complete arcs in the region $7057<q\le7577$.

We denote by $\overline{t}(3,q)$ the size of the smallest known complete arc in $\PG(3,q)$. The arcs obtained in \cite{BDMP-arXivR3_2018,DMP-R=3Redun2019,DMP-ICCSA2020,DavOst-IEEE2001} and in this paper (one arc of \cite{Sonino} is used also) give the value of $\overline{t}(3,q)$ providing the following theorem, cf. \eqref{eq2:R=3_r=4}.

\begin{theorem}
In the projective space $\PG(3,q)$, for the size $\overline{t}(3,q)$ of the smallest known complete arc and the smallest size $s_q(3,2)$ of a $2$-saturating set the following upper bound holds:
\begin{equation}\label{eq7:comp_bnd}
s_q(3,2)\le\overline{t}(3,q)\le c_4\sqrt[3]{q\ln q},~
~c_4= \left\{\begin{array}{@{}l@{}}
  2.61 \text{ if }  13\le q\le4373 \smallskip\\
  2.65\text{ if } 4373< q\le7577
\end{array}
  \right..
\end{equation}
\end{theorem}
In Figure \ref{fig1}, the sizes $\overline{t}(3,q)$ of the smallest known \cite{BDMP-arXivR3_2018,DMP-R=3Redun2019,DMP-ICCSA2020,DavOst-IEEE2001,Sonino} complete arcs in $\PG(3,q)$ divided by $\sqrt[3]{q\ln q}$ are shown by the bottom curve. The upper bounds of Theorem \ref{th6:t=1} and Theorem \ref{th3:infin_fam} for $R=3$, $t=1$ (also divided by $\sqrt[3]{q\ln q}$) are given by the top curve; the value $\lambda=3$ is used, see Table \ref{tab1}.

\begin{figure}[htbp]
\includegraphics[width=\textwidth]{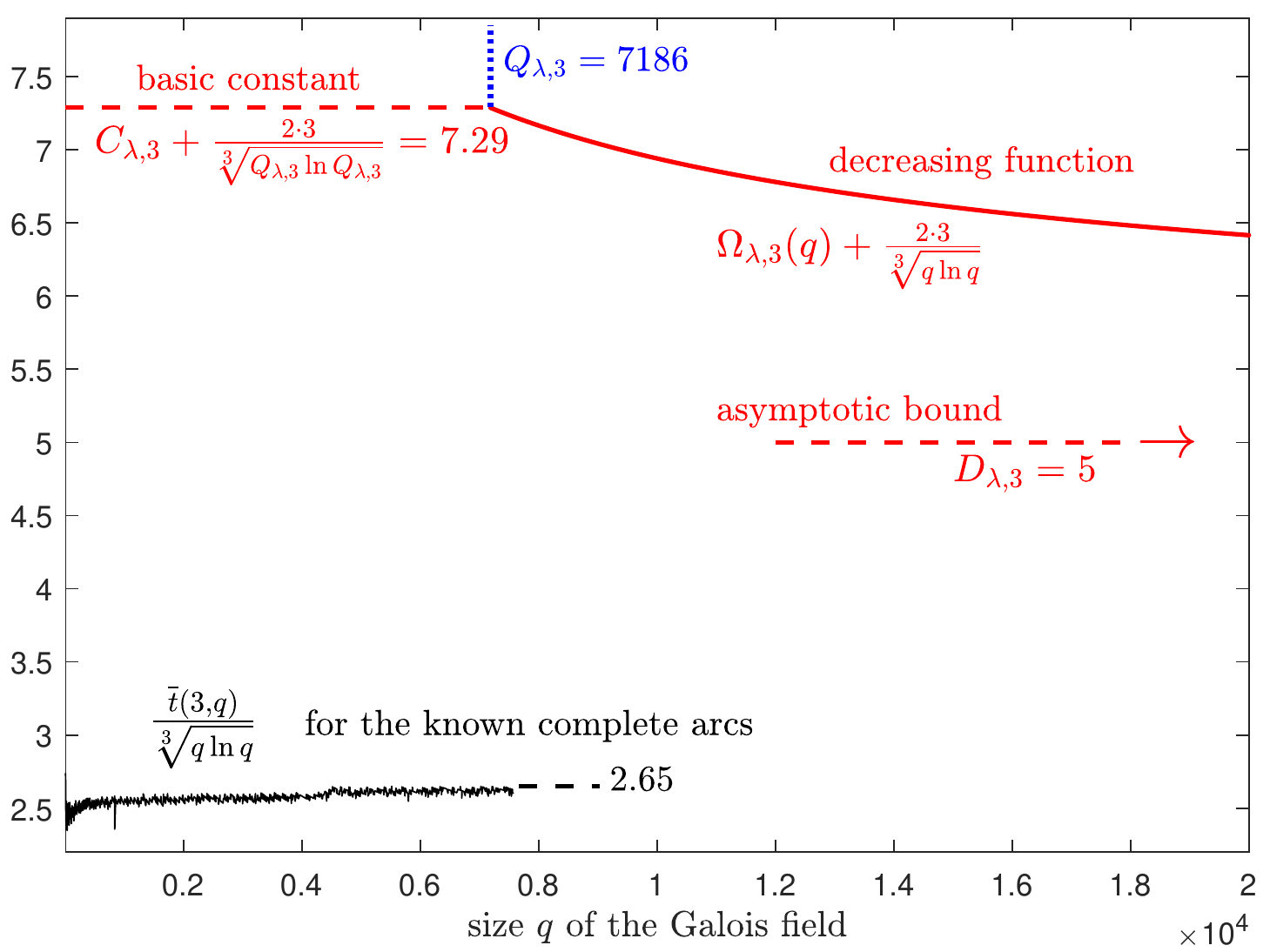}
\caption{The sizes $\overline{t}(3,q)$ of the smallest known complete arcs in $\PG(3,q)$ (\emph{bottom curve}) vs upper bounds of Theorem \ref{th6:t=1} and Theorem \ref{th3:infin_fam} (\emph{top curve}) for $R=3$, $t=1$; the sizes and bounds are divided by $\sqrt[3]{q\ln q}$; $\lambda=3$}
\label{fig1}
\end{figure}

In Figure \ref{fig1}, we see that the computer bounds are better than the theoretical ones, but the order of the value of the bounds is the same.

\section{Upper bounds on the length function $\ell_q(tR+1,R)$}\label{sec:qmconc}
Proposition \ref{prop8_lift_r=Rt+1} is a variant of the lift-constructions ($q^m$-concatenating constructions) for covering codes \cite{Dav90PIT,Dav95,DGMP-Bulg2008,DGMP-AMC,DavOst-IEEE2001,DavOst-DESI2010}, \cite[Section~5.4]{CHLS-bookCovCod}.

\begin{proposition}\label{prop8_lift_r=Rt+1}
\cite[Section 3]{Dav95}, \cite[Section 2, Construction $QM_1$]{DGMP-AMC}
Let an $ [n_{0},n_{0}-r_0]_{q}R$ code with $n_0\le q+1$ exist. Then there is an infinite family of
$[n,n-r]_{q}R$ codes with parameters
\begin{equation*}
  n=n_{0} q^m+R\theta_{m,q},~  r=r_0+Rm,~m\ge1.
\end{equation*}
\end{proposition}

\begin{proof}
  In terms of \cite{DGMP-AMC}, for its Construction $QM_1$, we take $\ell_0=0$ and use the trivial $(R,\ell_0)$-partition with $p_0=n_0$. Such a case is always possible.
\end{proof}

\begin{theorem}\label{th88:inf_fam}
Let $R\ge3$ be fixed. Let $q> q_0$ where $q_0$ is a certain constant. Let $f(q)$ be a decreasing function of $q$. Let $c_j$ be constants independent of $q$ such that $c_1\sqrt[R]{q\ln q}+c_2\le q+1$ and $f(q_{0})\sqrt[R]{q\ln q}+c_2\le q+1$. Let an $[n_0,n_0-(R+1)]_qR$ code exist with $n_0=\varphi\sqrt[R]{q\ln q}+c_2$ where $\varphi\in\{c_1,f(q)\}$. Then there is an infinite family of
$[n,n-r]_{q}R$ codes with growing codimension $r=tR+1$, $t\ge2$, and the following parameters:
\begin{equation}\label{eq8:inf_fam}
\begin{split}
n&=\varphi q^{(r-R)/R}\cdot\sqrt[R]{\ln q}+c_2q^{t-1}+R\theta_{t-1,q}\\
&<\left(\varphi+\frac{c_2+Rq/(q-1)}{\sqrt[R]{q\ln q}}\right)q^{(r-R)/R}\cdot\sqrt[R]{\ln q}\\
&<\left(\varphi+\frac{c_2+Rq_0/(q_0-1)}{\sqrt[R]{q_0\ln q_0}}\right)q^{(r-R)/R}\cdot\sqrt[R]{\ln q},\\
&  r=tR+1,~t\ge2,~q> q_0,~ R\ge 3.
\end{split}
\end{equation}
\end{theorem}

\begin{proof}
Using Proposition \ref{prop8_lift_r=Rt+1} with $r_0=R+1$, we obtain
  \begin{equation*}
  \begin{split}
n&=(\varphi\sqrt[R]{q\ln q}+c_2)q^m+R\theta_{m,q}=q^m\sqrt[R]{q\ln q}\left(\varphi+\frac{c_2+R\sum_{j=0}^mq^{-j}}{\sqrt[R]{q\ln q}}\right)\\
&=q^{t-1}\sqrt[R]{q\ln q}\left(\varphi+\frac{c_2+R\sum_{j=0}^{t-1}q^{-j}}{\sqrt[R]{q\ln q}}\right)<q^{t-1}\sqrt[R]{q\ln q}\left(\varphi+\frac{c_2+R\frac{q}{q-1}}{\sqrt[R]{q\ln q}}\right)
  \end{split}
  \end{equation*}
  where for the last inequality we use the sum of a geometric progression.
  Note that $r=tR+1=R+1+mR$, hence $m=t-1=(r-R-1)/R$. Finally, it is obvious that $(c_2+Rq/(q-1))/\sqrt[R]{q\ln q}$ is a decreasing function of $q$.
\end{proof}

Now we are able to prove Theorem \ref{th3:infin_fam}.

\begin{proof} [Proof of Theorem \ref{th3:infin_fam}]
For $t=1$, Theorem \ref{th3:infin_fam} is equivalent to Theorem \ref{th6:t=1}.

Let us consider the cases $t\ge2$.
We apply \eqref{eq8:inf_fam} with $q_0=Q_{\lambda,R}$.
\begin{description}
  \item[(i)] We use Theorem \ref{th6:t=1}(i)  and relation \eqref{eq8:inf_fam} of Theorem \ref{th88:inf_fam} taking $\varphi=f(q)=\Omega_{\lambda,R}(q)$, $c_2=2R$.
Also, for $q> Q_{\lambda,R}$, we have $\Omega_{\lambda,R}(q)<C_{\lambda,R}$, see Lemma \ref{lem7:properties}(iii) with \eqref{eq7:Omega<2}.

  \item[(ii)] The assertion follows from Theorem  \ref{th6:t=1}(i)(ii) and relation \eqref{eq8:inf_fam} with $\varphi=c_1\in\{C_{\lambda,R},\,\Omega_{\lambda,R}(Q_0)\}$, $c_2=2R$.

  \item[(iii)]  We use Theorem \ref{th6:t=1}(iii) and relation \eqref{eq8:inf_fam} with $\varphi=c_1\in\{D^{\min}_R,D_{\lambda,R}\}$, $c_2=2R$. This gives \eqref{eq3:boundConstAsym}.

       In the last inequality of \eqref{eq3:boundConstAsym}, we put $c=D^{\min}_R$ and consider $D^{\min}_R/R+\psi(q,R)$ where $\psi(q,R)=(2+q/(q-1))/\sqrt[R]{q\ln q}$. Obviously, $\psi(q,R)$ is a decreasing function of $q$.

       We check by computer that $D^{\min}_R/R+\psi(q,R)<3.35$ if $R<178$ and $q\ge41$.

       It can be shown that for a fixed $q$, we have $\lim_{R\rightarrow\infty}\psi(q,R)=2+q/(q-1)$ that gives rise to $\psi(q,R)<2+q/(q-1)$ as $\psi(q,R)$  is an increasing function of~$R$. So, for any $R\ge3$, we have $\psi(q,R)<3.03$ if $q\ge41$. Now we take into account that $D^{\min}_R<0.4R$ if $R\ge178$, see  \eqref{eq3:Dmin(N)<}. As a result, $D^{\min}_R/R+\psi(q,R)<3.43$.  \qedhere
       \end{description}
\end{proof}

\section*{Acknowledgments}
The research of S. Marcugini, and F. Pambianco was supported in part by the Italian
National Group for Algebraic and Geometric Structures and their Applications (GNSAGA -
INDAM) (Contract No. U-UFMBAZ-2019-000160, 11.02.2019) and by University of Perugia
(Project No. 98751: Strutture Geometriche, Combinatoria e loro Applicazioni, Base Research
Fund 2017-2019).


\begin{thebibliography}{99}
\bibitem{Hats-on-line}
\newblock S. Aravamuthan and S. Lodha,
Covering codes for Hats-on-a-line,
\emph{Electron. J. Combin.}, \textbf{13} (2006), \#21. \url{https://doi.org/10.37236/1047}

\bibitem{Bulev}
J. T. Astola and R. S. Stankovi\'{c},
Determination of sparse representations of multiple-valued logic functions by using covering codes,
\emph{J. Mult.-Valued Logic Soft Comput.}, \textbf{19} (2012), 285--306.

\bibitem{BDGMP-R2Redun2016}
D. Bartoli, A. A. Davydov, M. Giulietti, S. Marcugini and F. Pambianco,
New upper Bounds on the smallest size of a saturating set in a projective plane,
in \emph{2016 XV Int. Symp. Problems
of Redundancy in Inform. and Control Systems (REDUNDANCY), IEEE}, St. Petersburg,
Russia, Sept. 2016, 18--22. \url{https://doi.org/10.1109/RED.2016.7779320}

\bibitem{BDGMP-R2Castle}
D. Bartoli, A. A. Davydov, M. Giulietti, S. Marcugini and F. Pambianco,
New bounds for linear codes of covering radius 2,
in \emph{Coding Theory and Applications (Proc. 5th Int. Castle Meeting, ICMCTA 2017)}
(eds. \'{A}. I. Barbero, V. Skachek and {\O}. Ytrehus), Vihula, Estonia, Aug. 2017,
Lecture Notes in Computer Science, Springer, Cham, \textbf{10495} (2017), 1–10.
\url{https://doi.org/10.1007/978-3-319-66278-7_1}

\bibitem{BDGMP-R2R3CC_2019}
D. Bartoli, A. A. Davydov, M. Giulietti, S. Marcugini and F. Pambianco,
New bounds for linear codes of covering radii 2 and 3,
 \emph{Crypt. Commun.},  \textbf{11} (2019), 903--920. \url{https://doi.org/10.1007/s12095-018-0335-0}

  \bibitem{BDMP-arXivR3_2018}
D. Bartoli, A. A. Davydov, S. Marcugini and F. Pambianco,
Tables, bounds and  graphics of short linear codes with covering radius 3 and codimension 4 and
  5, preprint, arXiv:1712.07078 [cs.IT] (2020) \url{https://arxiv.org/abs/1712.07078}

\bibitem{libroBierbrauer}
J. Bierbrauer,
\emph{Introduction to Coding Theory},
Chapman \& Hall, CRC Press, Boca Raton, FL, 2005.

\bibitem{BierbStegan}
J. Bierbrauer and J. Fridrich,
Constructing good covering codes for applications in steganography,
in \emph{Transactions of Data Hiding Multimedia Security
III} (ed. Y. Q. Shi), Lecture Notes in Computer Science, Springer-Verlag, \textbf{4920}
(2008), 1--22. \url{https://doi.org/10.1007/978-3-540-69019-1_1}

\bibitem{Bo-Sz-Ti}
E. Boros, T. Sz\H{o}nyi and K. Tichler,
 On defining sets for projective planes,
\emph{Discrete Math.}, \textbf{303} (2005), 17-31. \url{https://doi.org/10.1016/j.disc.2004.12.015}

\bibitem{BoseBurt}
R. C. Bose, R. C. Burton,
A characterization of flat spaces in a finite geometry and the uniqueness of the Hamming and McDonald codes,
\emph{J. Comb. Theory}, \textbf{1} (1966), 96--104. \url{https://doi.org/10.1016/S0021-9800(66)80007-8}

\bibitem{Handbook-coverings}
R. A. Brualdi, S. Litsyn and V. Pless,
 Covering radius,
 in \emph{Handbook of Coding Theory} (eds. V. S. Pless and W. C. Huffman), vol. 1,
  Elsevier, Amsterdam (1998), 755--826.

  \bibitem{ChenListDec}
H. Chen, List-decodable codes and covering codes, preprint, arXiv:2109.02818 [cs.IT] (2021)
\url{https://arxiv.org/abs/2109.02818}

\bibitem{CHLS-bookCovCod}
G. Cohen, I. Honkala, S. Litsyn and A. Lobstein,
\emph{Covering Codes},
North-Holland Math. Library, \textbf{54}, Elsevier, Amsterdam, The Netherlands, 1997.

\bibitem{Lambert}
R. M. Corless, G. H. Gonnet, D. E. G. Hare, D. J. Jeffrey, D. E. Knuth,
On the Lambert $W$ function (PostScript),
\emph{Adv. Computat. Math.}, \textbf{5} (1996), 329–359. \url{https://doi.org/10.1007/BF02124750}


\bibitem{Dav90PIT}
A. A. Davydov,
Construction of linear covering codes,
\emph{Probl. Inform. Transmiss.}, \textbf{26}
(1990), 317--331. Available from: \url{http://iitp.ru/upload/publications/6833/ConstrCoverCodes.pdf}.

\bibitem{Dav95}
A. A. Davydov,
Constructions and families of covering codes and saturated sets of points in projective geometry,
\emph{IEEE Trans. Inform. Theory}, \textbf{41} (1995), 2071--2080. \url{https://doi.org/10.1109/18.476339}

\bibitem{DGMP-Bulg2008}
A. A. Davydov, M. Giulietti, S. Marcugini and F. Pambianco,
Linear covering codes over nonbinary finite fields,
in \emph{Proc. XI Int. workshop on algebraic and combintorial coding theory,
ACCT2008,} Pamporovo, Bulgaria, June 2008, 70–75.  Available from: \url{http://www.moi.math.bas.bg/acct2008/b12.pdf}.

\bibitem{DGMP-AMC}
A. A. Davydov, M. Giulietti, S. Marcugini and F. Pambianco,
Linear nonbinary covering codes and saturating sets in projective spaces,
\emph{Adv. Math. Commun.},  \textbf{5} (2011), 119--147. \url{https://doi.org/10.3934/amc.2011.5.119}

 \bibitem{DMP-R=tR2019}
A. A. Davydov, S. Marcugini and F. Pambianco,
New covering codes of radius $R$, codimension $tR$ and $tR+\frac{R}{2}$, and saturating sets in projective spaces,
 \emph{Des. Codes Cryptogr.} \textbf{87} (2019) 2771--2792. \url{https://doi.org/10.1007/s10623-019-00649-2}

 \bibitem{DMP-R=3Redun2019}
A. A. Davydov, S. Marcugini and F. Pambianco,
 New bounds for linear codes of covering radius 3 and 2-saturating sets in projective spaces,
in \emph{Proc. 2019 XVI Int. Symp. Problems Redundancy Inform. Control
Systems (REDUNDANCY)}, Moscow, Russia, Oct. 2019, IEEE Xplore, (2020) 47–52. \url{https://doi.org/10.1109/REDUNDANCY48165.2019.9003348}

 \bibitem{DMP-ICCSA2020}
A. A. Davydov, S. Marcugini and F. Pambianco,
Bounds for complete arcs in $\mathrm{PG}(3,q)$ and covering codes of radius 3, codimension 4, under a certain probabilistic onjecture,
in \emph{Computational Science and Its Applications – ICCSA 2020}, Lecture Notes in Computer Science, Springer, Cham, \textbf{12249} (2020), 107-122.
 \url{https://doi.org/10.1007/978-3-030-58799-4_8}

\bibitem{DavOst-EJC2000}
A. A. Davydov and P. R. J. \"{O}sterg{\aa}rd,
 On saturating sets in small projective geometries,
\emph{European J. Combin.}, \textbf{21} (2000), 563--570. \url{https://doi.org/10.1006/eujc.1999.0373}

\bibitem{DavOst-IEEE2001}
A. A. Davydov and P. R. J. \"{O}sterg{\aa}rd,
 Linear codes with covering radius $R=2,3$ and codimension $tR$,
\emph{IEEE Trans. Inform. Theory},  \textbf{47},
   (2001), 416--421. \url{https://doi.org/10.1109/18.904551}

\bibitem{DavOst-DESI2010}
A. A. Davydov and P. R. J. \"{O}sterg{\aa}rd,
 Linear codes with covering radius 3,
\emph{Des. Codes Cryptogr.}, \textbf{54} (2010), 253--271.  \url{https://doi.org/10.1007/s10623-009-9322-y}


\bibitem{denaux}
L. Denaux,
Constructing saturating sets in projective spaces using subgeometries, \emph{Des. Codes Cryptogr.}, to appear.
\url{https://doi.org/10.1007/s10623-021-00951-y}


\bibitem{denaux2}
 L. Denaux,
 Higgledy-piggledy sets in projective spaces of small dimension, preprint, arXiv:2109.08572 [math.CO] (2021)
\url{https://arxiv.org/abs/2109.08572}

\bibitem{EtzStorm2016}
T. Etzion and L. Storme,
 Galois geometries and coding theory,
\emph{ Des. Codes Cryptogr.},  \textbf{78} (2016), 311--350. \url{https://doi.org/10.1007/s10623-015-0156-5}

\bibitem{Identif}
G. Exoo, V. Junnila, T. Laihonen and S. Ranto,
 Constructions for identifying codes,
in \emph{Proc. XI Int. Workshop Algebraic Combin. Coding Theory, ACCT2008}, Pamporovo,
    Bulgaria, (2008), 92--98. Available from: \url{http://www.moi.math.bas.bg/acct2008/b16.pdf}.

\bibitem{GalKaba}
F. Galand and G. Kabatiansky,
Information hiding by coverings,
 in \emph{Proc. 2003 IEEE Inform. Theory Workshop (Cat. No. 03EX674)}, Paris, 2003, IEEE Xplore,
       (2003), 151--154. \url{https://doi.org/10.1109/ITW.2003.1216717}

\bibitem{GalKabaPIT}
F. Galand and G. Kabatiansky,
Coverings, centered codes, and combinatorial steganography,
 \emph{Probl. Inf. Transm.}, \textbf{45} (2009), 289--294. \url{https://doi.org/10.1134/S0032946009030107}

\bibitem{Cayley}
C Garcia and C. Peyrat,
Large Cayley graphs on an abelian group,
\emph{Discrete Appl. Math.}, \textbf{75} (1997), 125–133. \url{https://doi.org/10.1016/S0166-218X(96)00084-4}

\bibitem{Giul2013Survey}
M. Giulietti,
 The geometry of covering codes: small complete caps and saturating sets in Galois spaces,
in \emph{Surveys in Combinatorics 2013}, (eds. S. R. Blackburn, R. Holloway and  M. Wildon) London Math. Soc., Lecture
  Note Series, \textbf{409}, Cambridge University Press, Cambridge
  (2013), 51--90. \url{https://doi.org/10.1017/CBO9781139506748.003}

\bibitem{GrahSlo}
R. L. Graham and N. J. A. Sloane,
On the covering radius of codes,
\emph{IEEE Trans. Inform. Theory}, \textbf{31} (1985), 385--401. \url{https://doi.org/10.1109/TIT.1985.1057039}

\bibitem{LPN_CovCod}
Q. Guo, T. Johansson and C. L\"{o}ndahl,
Solving LPN using covering codes,
\emph{J. Cryptology}, \textbf{33}  (2020), 1--33. \url{https://doi.org/10.1007/s00145-019-09338-8}

\bibitem{HegNagy}
T. H\'{e}ger and Z. L. Nagy,
 Short minimal codes and covering codes via strong blocking sets in projective spaces, \emph{IEEE Trans. Inform. Theory}, to appear.
\url{https://doi.org/10.1109/TIT.2021.3123730}

\bibitem{Hirs_PG3q}
J. W. P. Hirschfeld,
\emph{Finite Projective Spaces of Three Dimensions},
Oxford Univ. Press, Oxford, 1985.

\bibitem{Hirs}
 J. W. P. Hirschfeld,
\emph{Projective Geometries over Finite Fields},
2$^{nd}$ edition, Oxford Univ. Press, Oxford, 1999.

\bibitem{HirsStor-2001}
J. W. P. Hirschfeld and L. Storme,
The  packing problem in statistics, coding theory and finite projective spaces: Update 2001,
 in (eds. A. Blokhuis, J. W. P. Hirschfeld, D. Jungnickel and J. A. Thas), \emph{Finite
    Geometries (Proc. 4th Isle of Thorns Conf., July 16-21, 2000)},
    Develop. Math.,  \textbf{3}, Kluwer, Dordrecht, 2001, 201--246. \url{https://doi.org/10.1007/978-1-4613-0283-4_13}

\bibitem{HirsThas-2015}
J. W. P. Hirschfeld and  J. A. Thas,
 Open  problems in finite projective spaces,
 \emph{Finite Fields   Appl.}, \textbf{32}(1) (2015), 44--81. \url{https://doi.org/10.1016/j.ffa.2014.10.006}

\bibitem{PartSumQuer}
C. T. Ho, J. Bruck and R. Agrawal,
Partial-sum queries in OLAP data cubes using covering codes,
\emph{IEEE Trans. Computers}, \textbf{47} (1998), 1326--1340. \url{https://doi.org/10.1109/12.737680}

\bibitem{HufPless}
W. C. Huffman and V. S. Pless,
\emph{Fundamentals of Error-Correcting Codes},
Cambridge Univ. Press, 2003.

\bibitem{KKKPS}
G. Kiss, I. Kov\'{a}cs, K. Kutnar, J. Ruff and P. \v{S}parl,
A note on a geometric construction of large Cayley graphs of given degree and diameter,
Studia Univ. \textquotedblleft
Babe\c{s} --Bolyai\textquotedblright , Mathematica,\
\textbf{LIV}(3) (2009), 77--84. Available from: \url{http://www.cs.ubbcluj.ro/~studia-m/2009-3/kiss.pdf}.

\bibitem{Klein-Stor}
A. Klein and L. Storme,
 Applications of finite geometry in coding theory and cryptography,
in \emph{Security, Coding
  Theory and Related Combinatorics, NATO Science for Peace and Security, Ser. -
  D: Information and Communication Security}, (eds. D. Crnkovi\'{c} and V. Tonchev) \textbf{29}, IOS Press
  (2011), 38--58. \url{https://doi.org/10.3233/978-1-60750-663-8-38}

\bibitem{swithc}
C. E. Koksal,
An analysis of blocking switches using  error control codes,
\emph{IEEE Trans. Inform. Theory}, \textbf{53} (2007), 2892--2897. \url{https://doi.org/10.1109/TIT.2007.901191}

\bibitem{Kovacs}
S. J. Kov\'{a}cs,
Small saturated sets in finite projective planes,
\emph{Rend. Math.}, Ser VII, \textbf{12} (1992), 157–164.

\bibitem{identif2}
T. K. Laihonen,
Sequences of optimal identifying codes,
\emph{IEEE Trans. Inform. Theory}, \textbf{48} (2002), 774--776. \url{https://doi.org/10.1109/18.986043}

\bibitem{LandSt}
I. Landjev and L. Storme,
Galois geometry and coding theory,
in \emph{Current Research Topics in Galois geometry,} (eds. J. De Beule and
  L. Storme), Chapter 8, NOVA Academic, New York, 2011, 187--214.

\bibitem{LenSerHat2002}
H. W. Lenstra, Jr., and G. Seroussi,
On hats and other covers,
in \emph{Proc. IEEE Int. Symp. Inform. Theory}, IEEE Xplore, Lausanne, Switzerland, July 2002 (2002), 342--342. \url{https://doi.org/10.1109/ISIT.2002.1023614}

\bibitem{LobstBibl}
A. Lobstein,
 Covering radius, an online bibliography (2019),
  Available from: \url{https://www.lri.fr/~lobstein/bib-a-jour.pdf}

\bibitem{MWS}
F. J. MacWilliams and N. J. A. Sloane,
\emph{The Theory of Error-Correcting Codes},
3$^{rd}$ edition, North-Holland, Amsterdam, The Netherlands, 1981.

\bibitem{MenNewCrypCovCod}
 B. Mennink and S. Neves,
 On the resilience of Even-Mansour to invariant permutations,
 \emph{Des. Codes Cryptogr.}, \textbf{89} (2021), 859--893.
\url{https://doi.org/10.1007/s10623-021-00850-2}

\bibitem{Nagy}
Z. L. Nagy,
 Saturating sets in projective planes and hypergraph covers,
 \emph{Discrete Math.} \textbf{341} (2018), 1078--1083. \url{https://doi.org/10.1016/j.disc.2018.01.011}

 \bibitem{Roth}
R. M. Roth,
\emph{Introduction to Coding Theory},
Cambridge, Cambridge Univ. Press, 2006.

\bibitem{Sonino}
A. Sonnino,
Transitive $PSL(2,7)$-invariant 42-arcs in 3-dimensional
projective spaces,
\emph{Des. Codes Cryptogr.}, \textbf{72} (2014), 455–463. \url{https://doi.org/10.1007/s10623-012-9778-z}

\bibitem{Struik}
R. Struik,
\emph{Covering Codes},
Ph.D thesis, Eindhoven University of Technology, The Netherlands, 1994. \url{https://doi.org/10.6100/IR425174}
\end{thebibliography}
 \end{document}